\documentclass[conference]{IEEEtran}
\IEEEoverridecommandlockouts
\usepackage{amsmath,amsfonts} % amssymb,
\usepackage{graphicx}
\usepackage{textcomp}
\usepackage{xcolor}
\usepackage{xspace}
\usepackage{soul}
\def\BibTeX{{\rm B\kern-.05em{\sc i\kern-.025em b}\kern-.08em
    T\kern-.1667em\lower.7ex\hbox{E}\kern-.125emX}}
    
\usepackage{subcaption}
\usepackage{amsthm}
\newtheorem{definition}{Definition}
\newtheorem{lemma}{Lemma}
\newtheorem{theorem}{Theorem}%[section]
\newtheorem{corollary}{Corollary}%[property]
\newtheorem*{problem}{Problem\textrm{ }Statement}

\usepackage{ragged2e} 
\usepackage{wrapfig}

\usepackage[most]{tcolorbox}
\newtcolorbox{myboxi}[1][]{
  breakable,
  title=#1,
  colback=gray!10!white,
  colbacktitle=white,
  coltitle=black,
  fonttitle=\bfseries,
  bottomrule=0.5pt,
  toprule=0.5pt,
  leftrule=0.5pt,
  rightrule=0.5pt,
  titlerule=0pt,
  colframe=black,
  boxsep=1pt,left=2pt,right=2pt,top=2pt,bottom=2pt
}

\usepackage{siunitx} % Required for alignment tables
\usepackage{multirow}
\usepackage{hyperref}

\usepackage{algorithm} 
\usepackage[noend]{algpseudocode} 
\usepackage{adjustbox}

\newcommand{\abacus}[0]{\textsc{Abacus}\xspace}
\newcommand{\parabacus}[0]{\textsc{ParAbacus}\xspace}

\newcounter{enum}
\newenvironment{packed_enum}{
\begin{list}{\textbf{(\arabic{enum})}}{
  \setlength{\itemsep}{0pt}
  \setlength{\parskip}{0pt}
  \setlength{\labelwidth}{-5 pt}
  \setlength{\leftmargin}{0 pt}
  \setlength{\itemindent}{0pt}
  \setlength{\topsep}{0pt}
  \usecounter{enum}}
}{\end{list}} 

\algnewcommand\algorithmicforeach{\textbf{for each}}
\algdef{S}[FOR]{ForEach}[1]{\algorithmicforeach\ #1\ \algorithmicdo}

\newcounter{MakisNOC}
\stepcounter{MakisNOC}

\newcounter{ZoiNOC}
\stepcounter{ZoiNOC}

\newcounter{JorgeNOC}
\stepcounter{JorgeNOC}

\newcounter{VaggosNOC}
\stepcounter{VaggosNOC}

\begin{document}

\title{Counting Butterflies in Fully Dynamic Bipartite Graph Streams}

\author{\IEEEauthorblockN{Serafeim Papadias$^{\#}$ \ \ Zoi Kaoudi$^{+}$ \ \ Varun Pandey$^{\#}$ \ \ Jorge-Arnulfo Quian\'e-Ruiz$^{+*}$\thanks{* The paper is dedicated in memory of Jorge; a mentor and a colleague who passed away so unexpectedly in May 2023.} \ \ Volker Markl$^{\#}$ \\ 
\{s.papadias, varun.pandey, volker.markl\}@tu-berlin.de$^{\#}$ \{zoka, joqu\}@itu.dk$^{+}$}
\IEEEauthorblockA{\textit{Technische Universit\"at Berlin$^{\#}$, ITU Copenhagen$^{+}$}
}}

%%%

% \author{Serafeim Papadias}
% \affiliation{\institution{Technische Universit\"at Berlin}\country{}}
% \email{s.papadias@tu-berlin.de}

% \author{Varun Pandey}
% \affiliation{\institution{Technische Universit\"at Berlin}\country{}}
% \email{varun.pandey@tu-berlin.de}

% \author{Zoi Kaoudi}
% \affiliation{\institution{ITU Copenhagen}\country{}}
% \email{zoka@itu.dk} 

% \author{Jorge-Arnulfo Quian\'e-Ruiz}
% \affiliation{\institution{ITU Copenhagen}\country{}}
% \email{joqu@titu.dk}

% \author{Volker Markl}
% \affiliation{\institution{Technische Universit\"at Berlin}\country{}}
% \email{volker.markl@tu-berlin.de}

\maketitle

\begin{abstract} 
A bipartite graph extensively models relationships between real-world entities of two different types, such as user-product data in e-commerce. 
Such graph data are inherently becoming more and more streaming, entailing continuous insertions and deletions of edges.  
A butterfly (i.e.,~$2 \times 2$ bi-clique) is the smallest non-trivial cohesive structure that plays a crucial role.
Counting such butterfly patterns in streaming bipartite graphs is a core problem in applications such as dense subgraph discovery and anomaly detection.  
Yet, existing approximate solutions consider insert-only streams and, thus, achieve very low accuracy in fully dynamic bipartite graph streams that involve both insertions and deletions of edges. 
Adapting them to consider deletions is not trivial either, because different sampling schemes and new accuracy analyses are required. 
We propose \abacus, a novel approximate algorithm that counts butterflies in the presence of both insertions and deletions by utilizing sampling. 
We prove that \abacus always delivers unbiased estimates of low variance.  
Furthermore, we extend \abacus and devise a parallel mini-batch variant, namely, \parabacus, which counts butterflies in parallel.  
\parabacus counts butterflies in a load-balanced manner using versioned samples, which results in significant speedup and is thus ideal for critical applications in the streaming environment.  
We evaluate \abacus/\parabacus using a diverse set of real bipartite graphs and assess its performance in terms of accuracy, throughput, and speedup. 
The results indicate that our proposal is the first capable of efficiently providing accurate butterfly counts in the most generic setting, i.e.,~a fully dynamic graph streaming environment that entails both insertions and deletions. 
It does so without sacrificing throughput, and even improves it with the parallel version. 
% \textcolor{blue}{
% Let's fortify the motivation of the introduction and the experiments in one shot. 
% Here is our chance to make it into TKDE I believe. 
% The practical applications of butterfly counting that we can experiment on are: 
%%%
% \begin{enumerate}
%     \item Social Network Analysis \begin{itemize}
%         \item Transitivity Ratio (TR)
%         \item Local Clustering Coefficient (LCC)
%         \item Global Clustering Coefficient (GCC) \makis{uses LCC} 
%         Find the equivalent metrics for community detection and apply them to bipartite graphs
%     \end{itemize}
%     \item Community Detection, focus on this
%     \item Anomaly Detection~\cite{tkdd20-thinkd, tkde21-cas}, Dense Subgraph Discovery~\cite{tkdd20-thinkd}
% \end{enumerate}
% }
\end{abstract}


\sloppypar
\section{Introduction}
\label{sec:introduction}

\noindent Bipartite graphs are a natural fit when it comes to modeling the relationship between two different types of entities in real-world applications~\cite{vldb19-xuemin, vldb18-fraud-alibaba}. 
For instance, Alibaba’s e-commerce platform models relationships between users and products via bipartite graphs~\cite{icde19-keynote-alibaba}. 
These graphs consist of billions of vertices (e.g., products, buyers, and sellers) and hundreds of billions of edges (e.g., representing clicks, orders, and payments). 
Nowadays, real-world bipartite graphs are inherently \textit{streaming}, entailing continuous insertions and deletions of vertices and edges~\cite{arxiv19-vasia} 
For example, Alibaba's user-product interactions are streams of very high velocity:
Reports of customer purchase activities specify that during a heavy period in 2017, 320 PB of log data were generated within only a six-hour period~\cite{tkdd22-sgrapp}. 
Consequently, it is vital to swiftly analyze the huge volume of incoming data and identify underlying trends or patterns in bipartite graph streams in order to gain valuable insights. 

The butterfly is the most basic substructure in bipartite graphs, similar to the triangle in unipartite graphs. 
A butterfly (i.e.,~$2\times2$ biclique) is a complete bipartite subgraph with two vertices belonging to one entity type and two vertices belonging to another entity type. 
Butterfly count is a metric that plays an important role in many applications. 
For instance, it is used to measure the  \textit{butterfly clustering coefficient} in a bipartite graph, which indicates how cohesive the graph is and can highlight how entities are clustered~\cite{lind2005cycles, aksoy2017measuring, opsahl2013triadic, robins2004small}.
This metric is important in many real-world applications, such as: in online recommendation systems to identify similar items~\cite{ijcai18-ahmed, icde19-jia, debs23-siachamis, edbt19-papadias}, cluster users, and enhance collaborative-filtering~\cite{DBLP:conf/ACMicec/MilitaruZ10}); in real-time anomaly detection~\cite{tkdd20-thinkd, tkde21-cas, kdd23-anomaly}; in fraud detection~\cite{vldb18-fraud-alibaba}.
%Note that measuring the per-vertex butterfly coefficients~\cite{OPSAHL2013159} is also important as they are used in assessing how clustered the immediate network of a single vertex is. 
Also, counting butterflies for each edge %of a bipartite graph is the first step towards
is required for the computation of $k$-bitrusses~\cite{wsdm18-peeling, icde20-wang, vldbj22-wang}, which is used in a varienty of applications, such as~community and spam detection~\cite{icde21-wang, sigmod21-tutorial, fang2020survey, vldb05-gibson, vldb20-kaixis, vldb23-wharf}.

Approaches that exactly count butterflies in static graphs~\cite{bigdata14-rectangle, vldbj22-wang, vldb19-xuemin, kdd18-tirthapura} are prohibitive for streams because they necessitate storing the whole graph in memory and take quadratic time to enumerate butterflies in the worst-case. 
Wang et al.~\cite{vldbj22-wang, vldb19-xuemin} 
devise a vertex-priority-based algorithm that considers both insertions and deletions of edges in a batch-dynamic setting. 
% \makis{State clearly that the vertex priority technique in ~\cite{vldb19-xuemin,vldbj22-wang} cannot be used for streaming graphs, as the author itself stated Of course, you cannot write about the author, thus, phrase it nicely.} \zoi{i think it's too much detail to have it here. better in  related work} \makis{ok, I will rewrite the related work better, as we already have this over there.} 
However, their per-batch computation mechanism cannot keep up with the pace of bipartite graph streams and results in stale counts in streaming applications. 
Approximate streaming methods that estimate the butterfly counts also exist~\cite{kdd18-tirthapura, cikm19-fleet, tkde21-cas}.     
% However, they all focus on insert-only bipartite graph streams,
% \textcolor{blue}{and none of them handles both edge insertions and deletions since their utilized sampling schemes cannot maintain uniform random samples in the presence of deletions. 
However, all of them are strictly focusing on insert-only %(emphasise insert-only) 
bipartite graph streams, and none of them can support both deletions and insertions of edges. 
This is primarily due to the inability of their sampling algorithms to maintain uniform random samples in presence of deletions. 
As a result, they fail to provide accurate butterfly counts in 
the most general and realistic setting, the fully dynamic one. 
This also directly results in a degradation of the output quality of many algorithms that rely on butterfly counts. 
Consider, for instance, the precision and recall metrics, which are utilized by anomaly detection algorithms to measure the quality of detected anomalies over time. 
Typically, an anomaly in bipartite graph streams appears when a certain number of butterflies that are formed is above some threshold~\cite{tkde21-cas, kdd23-anomaly, ranshous2015anomaly, sstd19-patroumpas}.   
Therefore, precision and recall will degrade significantly if the butterfly counts are maintained inaccurately, which will happen if edge deletions are ignored and not treated accordingly. 
In order to improve the quality of anomaly detection, 
% which is reflected in increasing the precision and recall metrics, 
it is vital to address the edge deletions appropriately.  

Counting butterflies in fully dynamic bipartite graph streams is challenging for three main reasons:  
(i) Due to the complexity of the butterfly counting problem and the nature of bipartite graphs, exact algorithms are prohibitive as they require the whole graph to be stored in main memory. 
(ii) An approximate solution is more suitable for achieving high throughput and maintaining a small memory footprint, but should provably deliver unbiased estimations. 
Doing so by using simple sampling techniques that only work for insert-only streams (e.g.,~ reservoir sampling) falls short for a bipartite graph stream that also involves deletions; 
(iii) Devising a parallel algorithm for increased throughput is appealing for graph streams but non-trivial. 
Specifically, all threads should have almost the same workload to not introduce stragglers while at the same time reducing contention among them.

We propose \abacus, which tackles all the above-mentioned challenges by providing accurate estimates for the butterfly counts in bipartite graph streams with deletions. 
Specifically, we use Random Pairing (RP)~\cite{vldbj08-gemulla} to maintain a uniform random sample of bounded size from a fully dynamic graph stream over which we estimate the butterfly counts. 
We prove that our estimates are unbiased and have low variance.  
Importantly, \abacus eliminates the need for extra hashmaps that bookkeep wedges~\cite{vldbj22-wang, kdd18-tirthapura} by refining its butterfly counts via set intersection operations. 
% \makis{vertex ordering for duplicate counting...}
% \makis{Can't the vertex ordering methods be used for dynamic graphs \textbf{to improve efficiency}? we do streaming...for dynamic it has been done~\cite{vldbj22-wang}}
% \makis{Be careful how you will state it, we do not want the silly reviewers to ask us to do further experiments with~\cite{vldbj22-wang} on dynamic graphs} 
Furthermore, we present a parallel variant of \abacus, called \parabacus, which processes a graph stream in mini-batches~\cite{zaharia2010spark, koliousis2016saber} using all available threads. 
More precisely, we maintain a versioned sample, which incorporates different states of the maintained sample that correspond to different edges in the mini-batch. 
Finally, we conduct the per-edge butterfly counting attributed to each edge in a mini-batch in parallel to enhance throughput.

In summary, we make the following major contributions: 

\begin{packed_enum}
    \item We formalize the problem of counting butterflies in fully dynamic streams, entailing both insertions and deletions of edges (Section~\ref{sec:problem-statement}). 
	
    \item We present \abacus, an approximate algorithm handling bipartite graph streams with both insertions and deletions. 
    We present a set intersection-based process for updating the butterfly counts that makes our algorithm more space-efficient as it alleviates the need for bookkeeping. 
    We refine our estimates by calculating the probability that a butterfly is formed between an edge that arrives and the sample we maintain (Section~\ref{sec:abacus-method}). 
    
    \item 
    We prove that \abacus always maintains unbiased estimates for butterfly counts of low variance. 
    We provide the time and space complexity of our algorithm     (Section~\ref{sec:accuracy-and-complexity}).
	
    \item We present the parallel version of \abacus, namely, \parabacus, which processes the graph stream in mini-batches using all available threads to enhance throughput. 
    Specifically, we maintain sample versions and distribute the counting of per-edge butterflies attributed to each edge in a mini-batch to all available threads (Section~\ref{subsec:parabacus}). 
	
    \item We conduct comprehensive experiments on a variety of real-world bipartite graph workloads. 
    We show that \abacus and \parabacus give both efficient and accurate estimates, which align with our theoretical analysis (Section~\ref{sec:experiments}). 
\end{packed_enum}

\noindent We then discuss related work in Section~\ref{sec:related-work}, where we stress that existing butterfly counting algorithms for graph streams fail to address the above-mentioned challenges attributed to edge deletions. 
We conclude the paper in Section~\ref{sec:conclusions}. 
\section{Problem Statement}
\label{sec:problem-statement}
\noindent Let us first introduce the notation we use throughout the paper (shown in Table~\ref{tab:notation})
%\makis{should we explain all the symbols here?} \zoi{nope}
and the core definitions necessary to formally define the problem we address. 
We consider undirected and unweighted bipartite graphs without zero-degree vertices, and without duplicate edges.  
We begin by defining a fully dynamic bipartite graph stream as follows: 

\begin{definition} %[Fully Dynamic Bipartite Graph Stream]
\label{def:bipartite-graph-stream} 
A \textbf{fully dynamic bipartite graph stream} $\Pi$ is a sequence of elements ($e^{(1)}, e^{(2)},\dots$). 
Let $G^{(t)} = (V^{(t)}, E^{(t)})$ be a bipartite graph that contains all edges $E^{(t)}$ that appear in the sequence $\Pi$ up to time $t$ (inclusive), and let their corresponding set of vertices $V^{(t)} = L^{(t)} \cup R^{(t)}$, which is separated into two disjoint partitions; a \textit{left} one, $L^{(t)}$, and a \textit{right} one, $R^{(t)}$, where $L^{(t)} \cap R^{(t)} = \emptyset$. 
It holds that $E^{(t)} \subseteq L^{(t)} \times R^{(t)}$. 
Also, $N_{v}^{(t)} = \{w \in V^{(t)} | (v,w) \in E^{(t)}\}$ denotes the set of neighbours of a vertex $v \in V^{(t)}$. 
For each discrete timestamp $t \geq 0$, let $e^{(t)} = (\{u^{(t)},v^{(t)}\},\delta)$ be the $t^{th}$ element in the sequence $\Pi$, where $\{u^{(t)},v^{(t)}\}$ is the actual edge and $\delta \in \{+,-\}$ denotes the change in $G$ at time $t$, i.e.,~whether the edge is inserted or deleted. 
More precisely, $(\{u^{(t)},v^{(t)}\},+)$ signifies that the edge did not exist up to time $t-1$, i.e.,~$(\{u^{(t-1)},v^{(t-1)}\} \notin E$, but will be inserted at time $t$, i.e.,~$(\{u^{(t)},v^{(t)}\} \in E$. 
Similarly, $\{u^{(t)},v^{(t)}\},-)$ indicates that an existing edge $\{u^{(t)},v^{(t)}\} \in E$ is about to get deleted. 
\end{definition}

\begin{figure}[t]
\centering
    \begin{subfigure}[t]{0.49\columnwidth}
        \includegraphics[width=\textwidth]{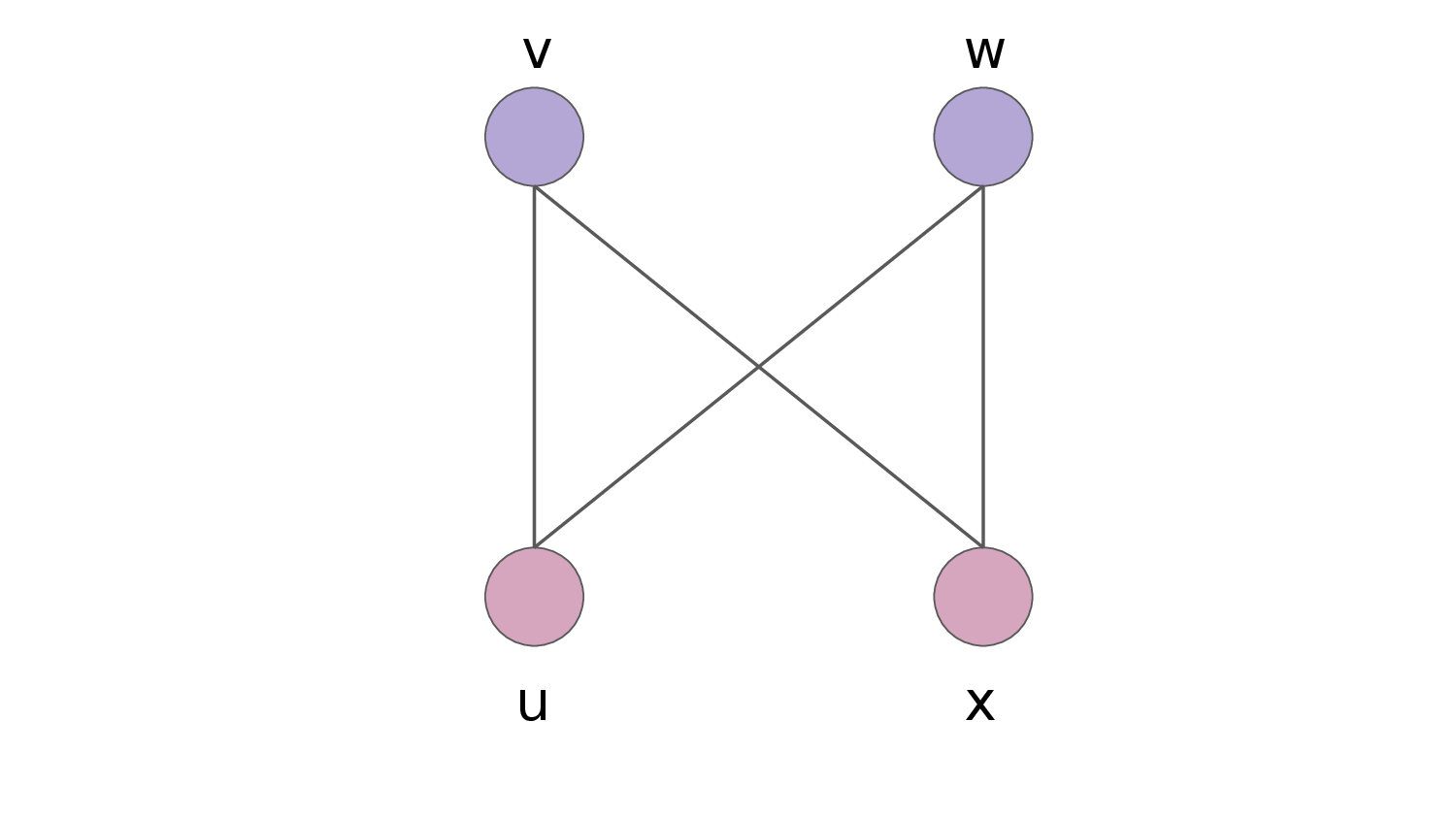}    
        \caption{Butterfly}
        \label{fig:butterfly}
    \end{subfigure}
    \begin{subfigure}[t]{0.49\columnwidth}
        \includegraphics[width=\textwidth]{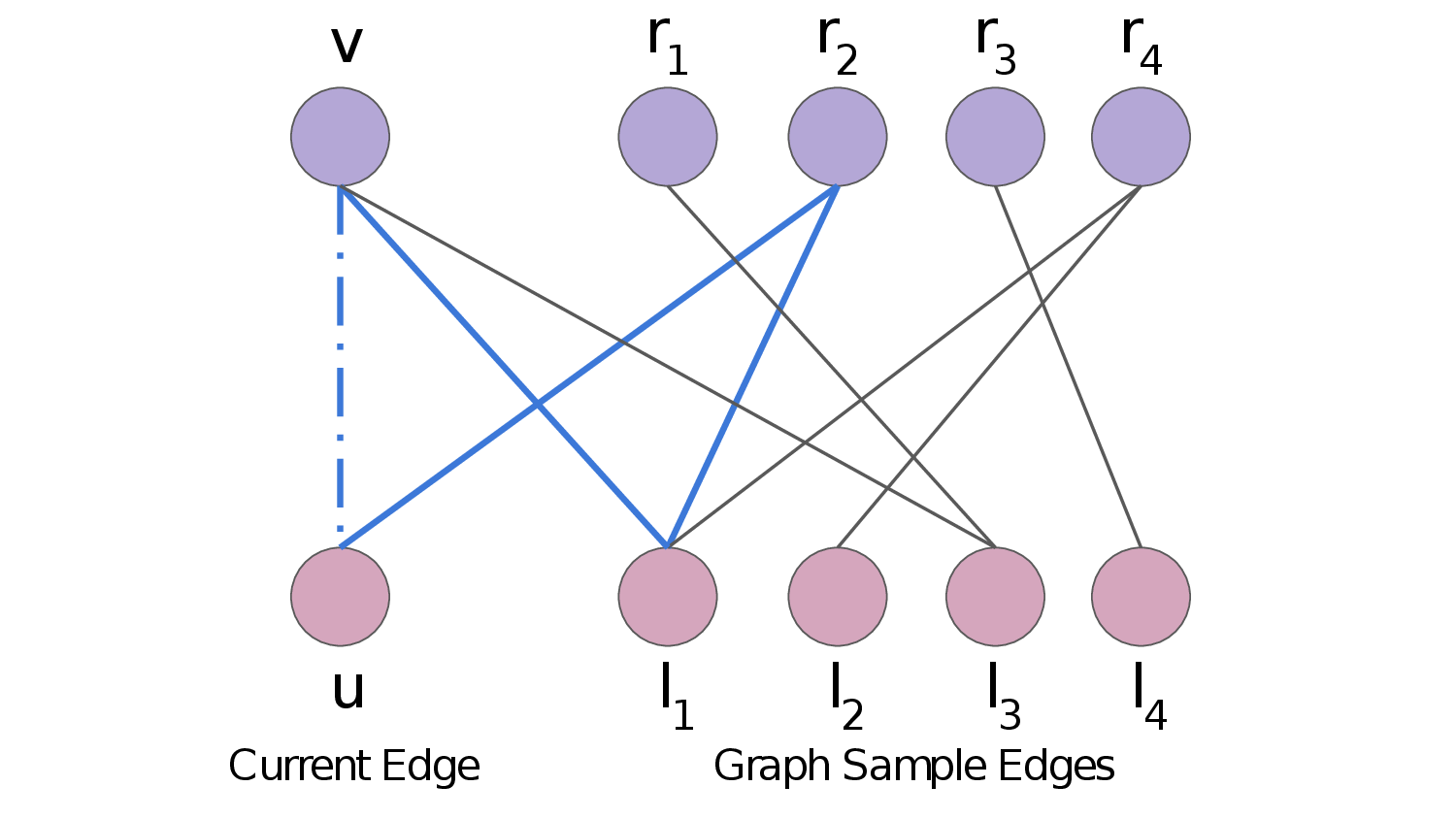}    
        \caption{Running Example}
        \label{fig:abacus-running-example}
    \end{subfigure} 
    \caption{{\bf (a)}~Butterfly structure. 
        {\bf (b)}~Running example showing the graph sample $\mathcal{S}$, an incoming edge $\{u, v\}$, and the butterflies that $\{u, v\}$ potentially forms with the edges in $\mathcal{S}$.} 
    \label{fig:abacus-example} 
\end{figure}

\noindent Implicitly, we assume that only new edges can be inserted (i.e.,~multigraphs with parallel edges are out of scope) and only edges that already exist can be deleted. 
Also, vertices that end up with degree zero, are deleted from $V^{(t)}, \forall t$. 
A butterfly is a complete $2\times2$ bipartite subgraph, where two vertices of one bipartition are connected with two vertices of the other bipartition.
For instance, a butterfly subgraph is depicted in Figure~\ref{fig:butterfly}. 
Let us now formally define a butterfly pattern that appears in such bipartite graphs as follows:

\begin{definition} %[Butterfly]
\label{def:butterfly} 
Given a bipartite graph $G^{(t)}$ and four vertices $u, v, w, x \in V^{(t)}$, where $u,x \in L^{(t)}$ and $v, w \in R^{(t)}$, a butterfly is the induced subgraph $\{u,v,w,x\}$
that is formed by the edges $(u,v),(u,w),(v,x),(w,x) \in E^{(t)}$. 
\end{definition}

% \hspace*{-1cm} 
\begin{table}[t!] % t!
    % \centering
    % \begin{center}        
        \caption{Notations.} 
        % \hskip-0.7cm
        \begin{adjustbox}{left}
            \hspace*{-0.3cm} % Adjust the -5mm to suit your needs
            \begin{tabular}{l|c}
                % l, left-aligned
                % r, right-aligned
                % c, center-aligned
                % S, for aligning decimal numbers (for siunitx package)
                \hline
                
                $\Pi$ & fully dynamic bipartite graph stream \\
    
                \hline
                
                $G^{(t)} = (V^{(t)}, E^{(t)})$ & bipartite graph at time $t$ 
                \\
                \hline
                
                $V^{(t)} = L^{(t)} \cup R^{(t)}$ & set of vertices (left and right partition) \\
                \hline
    
                $\delta$ & edge insertion $(\delta = +)$ or deletion $(\delta = -)$ \\
                \hline
    
                $\{u^{(t)},v^{(t)}\}$ & edge between two vertices $u$ and $v$ at time $t$ \\
                \hline
                
                $e^{(t)} = (\{u^{(t)},v^{(t)}\},\delta)$ & an element of $\Pi$ at time $t$ \\
                \hline
                
                $\{u^{(t)},v^{(t)},w^{(t)},x^{(t)}\}$ & butterfly formed by vertices $u,v,w,x$ at time $t$ \\
                \hline
                
                $B^{(t)}$ & set of butterflies in $G^{(t)}$ at time $t$ \\ 
                \hline 
                
                $\mathcal{S}^{(t)}$ & sample of the stream $\Pi$ at time $t$ \\
                \hline
                
                $N_u^{(t)}$ & set of neighbours of vertex $u \in V^{(t)}$ in $G^{(t)}$  \\
                \hline 
    
                $N_u^{\mathcal{S},(t)}$ & set of neighbours of vertex $u$ in the sample $S^{(t)}$  \\
                \hline 
                
                $d_u^{(t)}$ & degree of vertex $u \in V^{(t)}$  \\
                \hline 
    
                $d^{\mathcal{S},(t)}_u$ & degree of vertex $u$ in the sample $\mathcal{S}^{(t)}$  \\
                \hline 
                
                %$\Tilde{c}$
                $c$ & 
                butterfly count estimate \\ 
                \hline
                
                $c_b$ & %number of 
                \#uncompensated (``bad'') deletions $\in \mathcal{S}^{(t)}$ \\
                \hline    
                
                $c_g$ & \#uncompensated (``good'') deletions $\notin \mathcal{S}^{(t)}$ \\
                \hline    
                
                $k$ & memory budget %, i.e.,~
                (max \#edges in $\mathcal{S}$) \\
                \hline 
                
                $\mathcal{C}^{(t)}$ & set of created butterflies up to time $t$ \\
                \hline 
                
                $\mathcal{D}^{(t)}$ & set of deleted butterflies up to time $t$ \\
                \hline
    
                $M$ & number of edges in a mini-batch \\
                \hline
                
                %\textit{Erd\H{o}s R\'{e}nyi}
                % \textit{Synthetic} & %$0.1$B-$100$B
                % $100$B & $1$M & $1$M & - & - \\
    
                $\alpha$ & percentage of edges that are deletions \\
                \hline
                
                \hline
            \end{tabular} 
            \label{tab:notation}
        \end{adjustbox}
  % \end{center}
  % \hspace{-1cm} % put the table a bit to the left
\end{table}

We now define the problem of estimating the number of butterflies in a fully dynamic graph stream $\Pi$ under infinite window semantics~\cite{book16-minos}. 
Let $B^{(t)}$ denote the set of all butterflies in a bipartite graph $G^{(t)}$.  
We, thus, formally define the problem we focus on, as follows:

\begin{problem} %[Global and Local Butterfly Counting in Fully Dynamic Graph Streams]
\label{def:problem-statement}
Given a fully dynamic bipartite graph stream ($e^{(1)}, e^{(2)},\dots$) comprised of a sequence of edge insertions and deletions in a bipartite graph $G$, we aim to maintain butterfly count estimates $|B^{(t)}|$, using bounded memory, such that the estimates are unbiased and the errors are minimized. 
\end{problem} 

\noindent Note that we assume the traditional data stream model where the changes in the input stream can be accessed only once in the given order unless they are explicitly stored in memory.  
\section{\abacus}
\label{sec:abacus-method} 

\noindent We now present \abacus, an algorithm designed to facilitate the efficient and accurate counting of butterflies in fully dynamic bipartite graph streams.
\abacus employs sampling to maintain a subset of the edges of bounded size, and estimates butterfly counts through the maintained sample. 
First, we give an overview of the general workflow of \abacus. 
Subsequently, we present the sampling scheme \abacus employs for maintaining a uniform random sample. 
Finally, we illustrate the exact methodology we use to refine butterfly estimates.

\begin{algorithm} [t] 
	\small
    \caption{\textsc{\abacus}}
    \label{alg:abacus}
 
    \begin{algorithmic}[1] %[1] enables line numbers
        \State \textbf{Input:} fully dynamic input bipartite graph stream $\Pi = (e^{(1)}, e^{(2)},\dots)$, graph sample $\mathcal{S}$, memory budget $k \geq 2$
     
        \State \textbf{Output:} %global 
        butterfly count estimate $c$  
        %$\Tilde{c}$, 
        %, \textcolor{orange}{local butterfly count estimates $c[u]$ for each vertex $u$}  
        \State $\mathcal{S} \gets \emptyset$, $|E| \gets 0$, $c \gets 0$, $c_b \gets 0$, $c_g \gets 0$

        \ForEach {element $e^{(t)} = (\{u,v\}, \delta)$ in $\Pi$}
        \State// Update the Butterfly Count
        % \State $\mathtt{UpdateButterflyCount}$($\{u,v\}, \delta$)
        \State $increment$ = $\frac{sgn(\delta)}{Pr(|E|, c_b,c_g)}$
        \Comment{$sgn(\delta)$ %returns 
        is the sign of $\delta$}
        
        % \If{$\sum_{x \in N_u}d_x < \sum_{x \in N_v}d_x$} choose $v$ 
        \If{$\sum_{x \in S_u}d_x < \sum_{x \in S_v}d_x$} choose $v$ 
        \Comment{Else, %pick 
        $u$}
        \EndIf 
        % \ForEach{vertex $w \in N_u \setminus v$} 
        \ForEach{vertex $w \in N^{\mathcal{S}}_u \setminus v$}
        \Comment{$u$'s neighbors $\in \mathcal{S}$}
        % \State $\mathcal{CN}$ = $\mathtt{SetIntersection}(N_w, N_v)$
        \State $\mathcal{CN}$ = $\mathtt{SetIntersection}(N^{\mathcal{S}}_w, N^{\mathcal{S}}_v)$ 
        \ForEach{vertex $x \in \mathcal{CN}$}
        \State %$\Tilde{c}$
        $c$ %, $c[u]$, $c[v]$, $c[x]$, $c[w]$ 
        += $increment$
        \EndFor 
        \EndFor 
        \State// Update the Sample
        \If{$\delta = +$} $\mathtt{InsertToSample}$($\{u,v\}$, $k$)
        \ElsIf{$\delta = -$} $\mathtt{DeleteFromSample}$($\{u,v\}$)
        \EndIf
        \EndFor
        
    \end{algorithmic} 
\normalfont
\end{algorithm}

\subsection{Overview}
\label{subsec:overview}
\noindent Let us now elaborate on \abacus's main workflow as shown in Algorithm~\ref{alg:abacus} (lines 4-14). 
It ingests a fully dynamic graph stream, element by element, and maintains a uniform random sample $\mathcal{S}$ of bounded size $k$ by utilizing the Random Pairing (RP)~\cite{vldbj08-gemulla}. 
Note that the memory budget $k$ is the maximum possible sample size in \abacus, and thus, we use the terms \textit{memory budget} and \textit{sample size} interchangeably throughout the paper. 
The RP sampling scheme is suitable for streams that contain both insertions and deletions and ensures that the sample is always uniform. 
Uniformity is a property that enables us to extrapolate our estimations to the whole graph stream. 
Under deletions, the usual sampling schemes, such as reservoir sampling~\cite{vitter1985random}, do not guarantee uniformity. 
The processing happens per element (line~4), and once processed, \abacus evicts it from the main memory. 
On a high level, for each incoming element, \abacus first refines the maintained butterfly count estimates and then updates the sample $\mathcal{S}$. 
Specifically, for each incoming edge (either edge insertion or deletion), \abacus finds all the butterflies that the edge forms with the edges in the sample $\mathcal{S}$ (lines~8-11) and updates the butterfly count (lines~6, 11). 
Due to the inherent complexity of the butterfly structure itself, it is challenging and important to spot the formed butterflies efficiently. 
One must also prove that the estimates provided %via the sample 
are unbiased and of low variance, and thus, accurate and robust. 
Afterwards, \abacus updates the sample $\mathcal{S}$ (lines~13-14) by essentially deciding whether to insert the edge $\{u,v\}$ to $\mathcal{S}$ (if $\delta = +$) or delete it from $\mathcal{S}$ (if $\delta = -$), where $\delta$ indicates an edge insertion or deletion.

\subsection{Counting Butterflies per Edge}
\label{subsec:per-edge-butterfly-counts}
\noindent We now explain how \abacus finds the number of butterflies an incoming edge forms with the edges in the sample, as shown in Algorithm~\ref{alg:abacus} (lines~4-11).  
As we refine our butterfly counts using every incoming edge,  irrespective of whether the edge is later included in the sample, the operation of per-edge butterfly counting must be as efficient as possible. 

For each incoming edge $\{u, v\}$, \abacus chooses to compute the butterflies formed using the vertices on $u$'s side that are neighbors of $v$, or using the vertices on $v$'s side that are neighbors of $u$ (line~7). 
\abacus chooses to explore vertices on the side of the incoming edge's endpoint that has the smallest cumulative degree. 
This is a common heuristic~\cite{kdd18-tirthapura, vldbj22-wang}, and allows for choosing the cheapest side to conduct the counting. 
In the context of \abacus, choosing the cheapest side leads to conducting cheaper set intersections using vertices that belong to the bipartition with the smaller cumulative degree. 
The fact that the complexity of a set intersection operation between two sets is the size of the smallest set, leads to improved performance.  
As we see in our running example in Figure~\ref{fig:abacus-running-example}, we choose the side of vertex $v$ or equivalently we conduct the counting using the neighbors of $u$ in $\mathcal{S}$ as they have the smallest cumulative degree. 
Specifically, in our example $u$ has only one neighbor in the sample $\mathcal{S}$, i.e.,~$r_2$ with degree $2$, whereas $v$ has two neighbors in $\mathcal{S}$, i.e.,~ $l_1$ and $l_2$ with cumulative degree equal to $5$. 

Consider that \abacus chooses $v$ (line~7). 
If so, it explores every neighbor $w$ of $u$ in the sample $\mathcal{S}$ (i.e.,~excluding $v$) (line~8). 
Subsequently, \abacus finds the common neighbors, $\mathcal{CN}$, as the result of the set intersection between the set of neighbors of $v$ and the set of neighbors of $w$ (line~9). 
The result, $\mathcal{CN}$, of a set intersection (if not empty) contains all the vertices that serve as the fourth vertex $x$ forming a butterfly along with the edge's endpoints $u$, $v$, and the current vertex $w$ that \abacus explores amongst the neighbors of $u$ (line~10-11). 
In case the result of the set intersection is empty, this means that no butterfly that contains the incoming edge $\{u,v\}$ is formed through the vertex $w$. 
In our running example in Figure~\ref{fig:abacus-running-example}, since \abacus chooses $v$, the only vertex that is neighbor of $u$ in $\mathcal{S}$ is $r_2$, belonging to the right bipartition (shown in the upper part). 
Vertex $v$ has neighbours 
%$N_v = \{u,l_1,l_2\}$
$N^{\mathcal{S}}_v = \{u,l_1,l_2\}$ and vertex $r_2$ has neighbors $N^{\mathcal{S}}_{r_2} = \{u, l_1\}$. 
We exclude vertex $u$ from the corresponding neighboring sets as $u \notin \mathcal{S}$, and we observe that their common neighbor, $l_1$, indicates that the butterfly $\{u,v,l_1,r_2\}$ has been formed. 
Finally, \abacus adjusts the butterfly count for each discovered butterfly with a certain \textit{increment}  (line~11), as we describe in the sequel.

\begin{algorithm} [t] 
	\small
    \caption{\textsc{Random Pairing~\cite{vldbj08-gemulla}}}
    \label{alg:random-pairing}
 
    \begin{algorithmic}[1] %[1] enables line numbers
        % \State \textbf{Input:} fully dynamic input bipartite graph stream $\Pi = (e^{(1)}, e^{(2)},\dots)$, graph sample $\mathcal{S}$, memory budget $k \geq 2$
     
        % \State \textbf{Output:} global butterfly count estimate $\Tilde{c}$, local butterfly count estimates $c[u]$ for each vertex $u$ 
        % \State $\mathcal{S} \gets \emptyset$, $|E| \gets 0$, $c_b \gets 0$, $c_g \gets 0$
	    
        \Procedure{$\mathtt{InsertToSample}$}{$\{u,v\}$, $k$} 
        \State $|E| \gets |E| + 1$
        \If{$c_b + c_g = 0$}
        \If{$|\mathcal{S}| < k$} $\mathcal{S} \gets \mathcal{S} \cup \{\{u,v\}\}$
        \ElsIf{$Bernoulli(\frac{k}{|E|}) = 1$}
        \State replace a random edge in $\mathcal{S}$ with $\{u,v\}$
        \EndIf
        \ElsIf{$Bernoulli(\frac{c_b}{c_b+c_g}) = 1$}
        \State $\mathcal{S} \gets \mathcal{S} \cup \{\{u,v\}\}$
        \State $c_b \gets c_b - 1$
        \Else \textbf{ }$c_g \gets c_g - 1$
        \EndIf
        \EndProcedure
	    
        \Procedure{$\mathtt{DeleteFromSample}$}{$\{u,v\}$}
        \State $|E| \gets |E| - 1$
        \If{$\{u,v\} \in \mathcal{S}$} 
        \State $\mathcal{S} \gets \mathcal{S} \setminus \{\{u,v\}\}$
        \State $c_b \gets c_b + 1$
        \Else \textbf{ }$c_g \gets c_g + 1$
        \EndIf
        \EndProcedure
        
    \end{algorithmic} 
\normalfont
\end{algorithm}  

\subsection{Random Pairing for Uniform Samples}
\label{subsec:abacus-sampling-mechanism}  
\noindent Intuitively, the larger the size of the sample $\mathcal{S}$ that \abacus maintains, the more accurate the butterfly count estimations. 
This is in consistence with what has been demonstrated for triangle count estimation methods on fully dynamic graph streams~\cite{tkdd20-thinkd, tkdd17-triest}. 
Therefore, to minimize the information loss, \abacus strives to keep as many elements in the sample $\mathcal{S}$ as possible within a predefined memory budget $k \geq 2$. 
Random Pairing (RP)~\cite{vldbj08-gemulla} is a sampling scheme that always maintains a uniform random sample containing at most $k$ edges, given a fixed memory budget $k$ and a fully dynamic bipartite graph stream. 
Initially, the sample $\mathcal{S}$ as well as the bipartite graph stream are empty.  
Note that we initialize the compensation counters, $c_b$ and $c_g$, to zero (Alg.~\ref{alg:abacus}, line~3). 
Assuming that there is a set $E$ of edges in the input bipartite stream that have not yet been deleted, we now describe the core functionality for inserting or removing an edge to or from $S$ that \abacus maintains, as illustrated in Algorithm~\ref{alg:random-pairing}. 
More precisely, when an edge deletion appears (lines~12-16), if the edge is in $\mathcal{S}$, then RP increases the $c_b$ (line~15); otherwise, it increases $c_g$ (line~16).  
Intuitively, the counters $c_b$ and $c_g$ signify the number of deletions that need ``compensation'' from upcoming insertions. 
% \makis{Here, is the first time I elaborate on the compensation counters of RP sampling scheme. Should we do it explicitly from the preliminaries section? See R7.O1} \zoi{where is the preliminaries section? we don't have one, right?} \makis{sorry, I meant the ``problem statement'' section.} 
On the other hand, when an edge insertion appears (lines~2-10) and there are no deletions to compensate for, i.e.,~$c_b+c_g = 0$ (line~17), \abacus processes the upcoming edge insertion as in reservoir sampling~\cite{vitter1985random} (lines~4-6).  
In particular, if \abacus has not exhausted the memory budget yet, i.e.,~$|\mathcal{S}| < k$, then we append the newly arrived edge to $\mathcal{S}$ (line~4); else, we replace a random edge in $\mathcal{S}$ with the new edge with a probability $k/|E|$ (lines~5-6).  
In case the sum of the compensation counters is not zero, then we consider them when calculating the probability of replacing an edge from $\mathcal{S}$ with the new one (line~7), and update the counters accordingly (lines~9-10). 
\noindent Note that the uniformity of $\mathcal{S}$ allows for exactly calculating the discovery probability of each butterfly in a deterministic way.

\subsection{Butterfly Count Update Mechanism}
\label{subsec:abacus-update-mechanism} 
\noindent We now elaborate on how \abacus updates its butterfly count estimates. 
As we described in Section~\ref{subsec:per-edge-butterfly-counts}, when an edge (insertion or deletion) arrives, \abacus first finds the butterflies that the edge forms with the edges in $\mathcal{S}$. 
Here, we show how \abacus utilizes the spotted butterflies to update the butterfly count estimate. 
We explain \textit{how much} \abacus modifies its butterfly count estimate, such that it always remains unbiased. 

In Algorithm~\ref{alg:abacus}, the specific increment amount with which \abacus refines the maintained butterfly count is important for providing accurate estimations (lines~6, 11). 
More precisely, each incoming edge contributes to the creation of some new butterflies if it is an insertion, or causes the deletion of some existing butterflies if it is a deletion. 
We can reason about the created or deleted butterflies with the edges that exist in the sample, yet for the butterfly counts in the whole graph stream, we can only extrapolate using the actual counts we gain through the sample.    
Furthermore, the fact that we are maintaining a uniform random sample implies that the created or deleted butterflies are discovered with a certain probability, which we can exactly calculate. 
Specifically, each time an element $e^{(t)} = (\{u,v\}, \delta)$ arrives (Algorithm~\ref{alg:abacus}, line~4), every created or deleted butterfly $\{u,v,w,x\}$, where $u,w \in L^{(t)}$ and $v,x \in R^{(t)}$, is successfully discovered (lines~8-11) if and only if three specific edges exist in the sample $\mathcal{S}$, namely, the edges $\{u,x\}$, $\{w,x\}$, and $\{v,w\}$. 
Assuming that $y = \min(k, |E^{(t)}|+c_g+c_b)$, which is the size of the sample $\mathcal{S}$, we prove that the above-mentioned butterfly discovery probability through the sample is as follows: 
\begin{equation}
    \label{eq:reciprocal}
    \hspace*{-1.0cm} 
    Pr(|E^{(t)}|,c_b,c_g) = \frac{y}{T} \cdot \frac{y-1}{T-1} \cdot \frac{y-2}{T-2}\text{, with } T = |E^{(t)}|+c_b+c_g
\end{equation} 
\noindent where $c_g$, $c_b$ are compensation counters for the edge deletions in the RP sampling, and $E^{(t)}$ 
is the set of edges that remain in the input bipartite streaming graph (without being deleted) after the $t$-th element of the stream is processed. 
For each butterfly that \abacus discovers using the sample, it updates the corresponding butterfly estimates with the reciprocal of the probability that the butterfly is discovered (line~6). 
Utilizing the reciprocal of the discovery probability as the amount of change per butterfly discovered makes the expected amount of changes in the estimated butterfly counts of \abacus exactly one, and thus, enables us to provide unbiased estimates. 
Therefore, for each incoming edge, \abacus calculates the reciprocal \textit{increment} based on Equation~\ref{eq:reciprocal} and uses it to refine the count estimates for the butterflies that it discovers with the incoming edge (line~6,11).

\section{Accuracy and Complexity}
\label{sec:accuracy-and-complexity}
\noindent We now study the accuracy of our algorithm and prove that \abacus consistently maintains unbiased estimates of low variance for the butterfly count. %, with expected values that are the same as the true counts.  
% Demonstrating the absence of bias in \abacus is very important, as it ensures reliable estimations and accurate representations of true butterfly counts. 
We then present the time and space complexity of our algorithm.
% \noindent We now study the accuracy of our algorithm and prove that \abacus consistently maintains unbiased estimates of the butterfly count, with expected values that are the same as the true counts.  
% Demonstrating the absence of bias in \abacus is very important, as it ensures reliable estimations and accurate representations of true butterfly counts. 
% We then present the time and space complexity of our algorithm.

\subsection{Accuracy Analysis}
\label{subsec:accuracy-analysis}

\noindent Before proving the unbiasedness of \abacus, we first provide the following lemma for the butterfly discovery probability.

\begin{lemma}[Butterfly Discovery Probability]
\label{lemma:discovery-probability-butterflies} 
In \abacus, any three distinct edges in the bipartite graph $G^{(t)} = (V^{(t)}, E^{(t)})$ are sampled with the probability shown in Equation~\ref{eq:reciprocal}. 
Therefore, assuming that $p^{(t)}$ is the probability of the Equation~\ref{eq:reciprocal} and $S^{(t)}$ is the sample of the bipartite graph, respectively, after the $t$-th element $e^{(t)}$ is processed by \abacus (Algorithm~\ref{alg:abacus}), then the following holds:   
\begin{align}
\label{eq:discovery-probability-butterflies}
    % \hspace{-1.5cm}
    Pr(\{u,v\} &\in \mathcal{S}^{(t)} \cap \{w,x\} \in \mathcal{S}^{(t)} \cap \{y,z\} \in \mathcal{S}^{(t)}) = p^{(t)}, 
    \nonumber \\ 
    \forall &t \geq 1, \forall \{u,v\} \ne \{w,x\} \ne \{y,z\} \in E^{(t)} 
\end{align}
\end{lemma}

\begin{proof}[Proof] 
See Appendix~\ref{proof:discovery-probability-butterflies}.
% \textcolor{blue}{See Appendix~\ref{proof:discovery-probability-butterflies} in \abacus' extended version~\cite{}.}
\end{proof}

\noindent We now formally prove that \abacus maintains unbiased butterfly count estimates as stated in the following theorem:

\begin{theorem}[Unbiasedness]
\label{theorem:unbiasedness}
\abacus provides unbiased butterfly count estimates at any point in time. 
Specifically, for Algorithm~\ref{alg:abacus} it holds:
\begin{align}
\label{eq:unbiasedness}
    \mathop{\mathbb{E}}(c^{(t)}) &= |B^{(t)}|, \forall t \geq 1 
\end{align}
\noindent where $|B^{(t)}|$ is the true butterfly count at time $t$. 
\end{theorem}

\begin{proof}[Proof]
See Appendix~\ref{appendix:proof-unbiasedness}.
% \textcolor{blue}{See Appendix~\ref{appendix:proof-unbiasedness} in \abacus' extended version~\cite{}.} 
\end{proof}

\begin{theorem}[Variance]
\label{theorem:variance} 
\abacus provides estimates of bounded variance at any point in time. Specifically, for Algorithm~\ref{alg:abacus} it holds: 
\begin{align*} 
\hspace{-0.3cm}
    Var[c] = \gamma \mathbb{E}[c] - \mathbb{E}[c]^2 + 2\gamma^2 (y_1 \frac{\binom{|E|-8}{k-8}}{\binom{|E|}{k}} + y_2 \frac{\binom{|E|-7}{k-7}}{\binom{|E|}{k}} + y_3 \frac{\binom{|E|-6}{k-6}}{\binom{|E|}{k}} ) 
\end{align*}
\noindent where 
% $p_{i,j}$ is the probability that the graph sample $\mathcal{S}$ contains both the butterfly $B_i$ and the $B_j$, the 
$y_1, y_2, y_3$ indicate how many pairs of butterflies that share $0$, $1$, and $2$ edges, respectively, exist in the graph at the current time $t$ (we omit time notation for simplicity), 
$c$ 
% \kipou{maybe $|\mathbb{E}[c]$?} \makis{correct.} 
% is the expected value for 
is the butterfly count estimation of \abacus, whose expected value equals the ground truth 
% \kipou{or keep the $c$ as is, and say "is the estimation of \abacus, whose expected value equals the ground truth..." instead of "is the expected value for the estimation"} \makis{correct. done.} 
% the estimation of \abacus, which equals to the ground truth
butterfly count, $|E|$ is the number of valid edges in the graph stream that have not yet been deleted,  $k$ is the memory budget %or equivalently the maximum number of edges in the sample 
of $\mathcal{S}$, and $\gamma = \binom{E}{k} / \binom{E-4}{k-4}$. % that \abacus maintains. 
% By upper bounding $\sum_{i,j}p_{i,j}$ as follows: 
% \begin{align} 
%     \sum_{i,j} p_{i,j} \leq \binom{E[c]}{2} \times \frac{\binom{|E|-6}{k-6}}{\binom{|E|}{k}} \nonumber
% \end{align}
A \textit{tight} upper bound of the variance is:  
\begin{align*}  
    Var[c] \leq \gamma \mathbb{E}[c] + 2 \gamma^2 \binom{E[c]}{2} \times \frac{\binom{|E|-6}{k-6}}{\binom{|E|}{k}} - \mathbb{E}[c]^2 
\end{align*} 
\end{theorem}

\begin{proof}[Proof]
See Appendix~\ref{appendix:proof-variance}. 
% \textcolor{blue}{See Appendix~\ref{appendix:proof-variance} in \abacus' extended version~\cite{}.} 
\end{proof}

%%%

\begin{corollary}[Concentration]
\label{corollary:concentration} 
\abacus provides estimates that concentrate around the expected value at any point in time. 
Specifically, for Algorithm~\ref{alg:abacus} and any constant $\lambda > 0$ it holds:
\begin{align}
    Pr[|c - \mathbb{E}[c]| \geq \lambda \times \sqrt{Var[c]}] \leq \frac{1}{\lambda^2} \nonumber
\end{align}
where $c$ is the butterfly estimate of \abacus. 
\end{corollary}

\begin{proof}[Proof] 
See Appendix~\ref{appendix:proof-concentration}. 
% \textcolor{blue}{See Appendix~\ref{appendix:proof-concentration} in \abacus' extended version~\cite{}.} 
\end{proof}

\subsection{Complexity Analysis}
\label{subsec:complexity-analysis}
\noindent Here, we analyze \abacus's complexity with respect to both time and space when processing a bipartite graph stream. 

\noindent \textbf{Time Complexity.} We provide the worst-case analysis in Theorem~\ref{theorem:time-complexity}. 
Essentially, we claim that given a fixed memory budget $k$, \abacus scales linearly with the number of elements in the input bipartite graph stream. 
The time complexity theorem is as follows:

\begin{theorem}[Time Complexity of \abacus]
\label{theorem:time-complexity}
Algorithm~\ref{alg:abacus} takes $O(k^2 t)$ time to process the first $t$ elements in the input bipartite graph stream, where $k$ is the maximum number of edges maintained in the sample. 
\end{theorem}

\begin{proof}[Proof]
\label{proof:time-complexity}
The most expensive operation in Algorithm~\ref{alg:abacus}, is the per-edge butterfly counting for spotting the butterflies that each incoming edge $e^{(t)} = (\{u^{(t)},v^{(t)}\},\delta)$ forms with the edges in the graph sample. 
% The per-edge counting process takes $\Lambda = O(\min \{\sum_{x \in N_u} \min \{d_x,d_v\}, \sum_{x \in N_v} \min \{d_w, d_u\}\})$ time for an incoming edge $\{u, v\}$. 
The per-edge counting process takes $\Lambda = O(\min \{\sum_{x \in N^{\mathcal{S}}_u} \min \{d_x,d_v\}, \sum_{x \in N^{\mathcal{S}}_v} \min \{d_w, d_u\}\})$ time for an incoming edge $\{u, v\}$. 
All the vertex degrees are upper-bounded by $k$, which is the maximum number of edges in the sample. 
Therefore, $\Lambda = O(k^2)$, and the time complexity for processing the first $t$ elements is $O(k^2 t)$. 
\end{proof}

\noindent \textbf{Space Complexity.} 
We provide the space analysis
in Theorem~\ref{theorem:space-complexity}. 
{In a nutshell, given a fixed memory budget $k$, \abacus needs to store at most $k$ elements as a sample of the input bipartite graph stream, and a counter for the butterfly count estimate. 
The theorem for the space complexity is as follows: 

\begin{theorem}[Space Complexity of \abacus]
\label{theorem:space-complexity} 
Algorithm~\ref{alg:abacus} has a space complexity of $O(k)$. 
\end{theorem}

\begin{proof}[Proof]
\label{proof:space-complexity}
While \abacus processes the first $t$ elements of the input bipartite graph stream maintains a single estimate for the butterfly count.  
Furthermore, \abacus maintains up to $k$ edges that consist of the graph sample, where $k$ is the memory budget and a parameter of our algorithm.  
Therefore, the space complexity for processing the first $t$ elements is $O(k)$. 
\end{proof}  
\begin{figure}[t]
\centering 
\includegraphics[width=0.42 \textwidth]{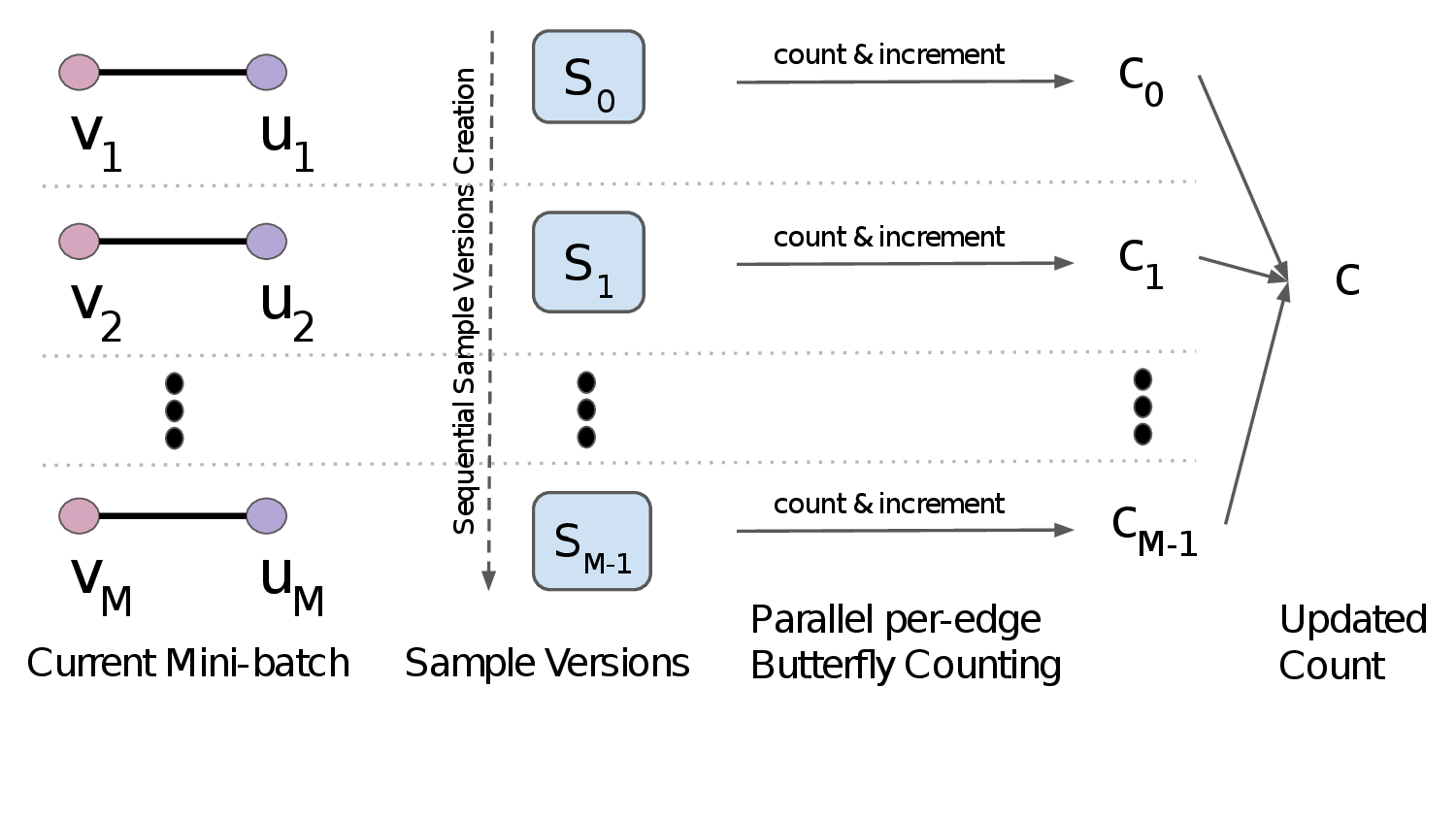}
 \vspace{-0.2cm}
\caption{Overview of \parabacus.} 
\label{fig:parabacus-overview} 
\end{figure}

\section{\parabacus: \abacus goes Parallel}
\label{subsec:parabacus}
\noindent We now present \parabacus, %which is the 
the parallel variant of \abacus that processes the %bipartite graph 
stream in mini-batches. 
The %primary
main challenge is to process the edges in a mini-batch simultaneously while attaining the same accuracy as \abacus. 
To achieve this, we revert \abacus's workflow, namely, we first %conduct the sample updates that
perform sample updates corresponding to each edge in the mini-batch and create versions of the sample. 
%Subsequently
Next, we conduct the per-edge butterfly counting between an edge %in the mini-batch 
and its corresponding version of the sample in parallel. 
Next, we first describe how \parabacus processes each mini-batch and, then, provide correctness proof, and its time and space complexities. 
% \noindent We now present \parabacus, which is the parallel variant of \abacus that processes the bipartite graph stream in mini-batches. 
% The primary challenge is to process the edges in a mini-batch simultaneously while attaining the same accuracy as \abacus. 
% To achieve this, we revert \abacus's workflow, namely, we first conduct the sample updates that correspond to each edge in the mini-batch and create versions of the sample. 
% Subsequently, we conduct the per-edge butterfly counting between an edge in the mini-batch and its corresponding version of the sample in parallel. 
% Next, we first describe how \parabacus processes each mini-batch and, then, provide correctness proof and its time and space complexities. 

\subsection{\parabacus Algorithm}
\label{subsec:parabacus-algorithm-description}

\noindent \textbf{Sampling and Versioned Samples.} 
\parabacus processes the bipartite graph stream in mini-batches of size $M$. 
Figure~\ref{fig:parabacus-overview} shows the overview of \parabacus' workflow.   
In specific, we process the graph stream in mini-batches that contain $M$ edges each, namely, $\{u_1,v_1\}$, $\{u_2,v_2\}$,$\dots$, $\{u_M,v_M\}$. 
Recall that in \abacus, for each incoming edge, we first conduct the per-edge butterfly counting to refine the butterfly count estimates and then trigger the sample update procedure. 
Since the per-edge butterfly counting is the most time-consuming operation in \abacus's workflow, we ought to effectively parallelize it when processing whole mini-batches. 
To achieve this, \parabacus sequentially processes the edges in the mini-batch once to calculate all the sample states, which would be created as if the edges were processed by \abacus. 
Each different state of the sample $S$ is a distinct sample version. 
For example, as shown in Figure~\ref{fig:parabacus-overview}, $\mathcal{S}_0$ indicates the state of the sample immediately after the arrival of the mini-batch. 
Assuming the edges arrive in the following sequence $\{u_1,v_1\}$, $\{u_2,v_2\}$,$\dots$,$\{u_M,v_M\}$, the edge $\{u_1,v_1\}$ will observe the $\mathcal{S}_0$ version of the sample. 
Subsequently, the edge $\{u_2,v_2\}$ would observe the version $\mathcal{S}_1$, which corresponds to the state of the sample after updating $\mathcal{S}_0$ with incorporating the edge $\{u_1,v_1\}$. 
Similarly, the edge $\{u_M,v_M\}$ would observe the  $\mathcal{S}_{M-1}$ version, which corresponds to the state of the sample after incorporating the updates due to all the edges in the mini-batch except the $M$-th one. 
\parabacus maintains all the calculated versions of the samples $S_0,S_1,\dots,S_{M-1}$ in a single \textit{versioned sample} data structure. 
In particular, we use adjacency lists to store the edges that we sample. 
More precisely, in the versioned sample, each vertex of the sample stores its neighbors in its adjacency list that might change between versions. 
However, from one version to another, we store only the discrepancies between the neighboring sets of each vertex to save space. 
Furthermore, along with each sample version we cache a triplet containing the following information: $\{s, c_g, c_b\}$, where $s$ is the number of edges in the graph stream and $c_g,c_b$ are the good and bad edge deletions that need compensation at the point of creation of the sample version. 
\parabacus utilizes that triplet to calculate the increment using which it refines the butterfly count estimates, as we described in Section~\ref{sec:abacus-method}.

\noindent \textbf{Parallel Per-edge Butterfly Counting.} After assembling the versioned sample,  we have every sample state $S_i$ where $i \in \{0,\dots,M-1\}$ readily available. 
This allows \parabacus to conduct the butterfly counting operations for all the $M$ edges in the mini-batch in parallel. 
Specifically, we have to conduct per-edge butterfly counting between each edge in the mini-batch and its corresponding sample version using a separate thread among the available ones. 
For example, in Figure~\ref{fig:parabacus-overview}, we have to count the butterflies formed between edge $\{v_1,u_1\}$ and $S_0$, the ones formed between $\{v_1,u_1\}$ and $S_1$, all the way to the ones formed between $\{v_{M-1},u_{M-1}\}$ and $S_{M-1}$. 
% \makis{R7.O1. $S_{something}$ is used for the sample versions in \parabacus. I must find another symbol for the neighbors of a vertex $u$ in the sample $\mathcal{S}$, namely, $\mathcal{S}_u$. Maybe $N^{\mathcal{S},(t)}_u$} 
Assuming that there are $p$ threads available and that we process mini-batches of $M$ edges where $p \leq M$, \parabacus groups the edges into $p$ equal-sized sets. 
Therefore, each thread receives a subset of edges from the mini-batch and has to count the butterflies that its corresponding edges form with their corresponding sample versions. 
Subsequently, every edge, e.g.,~ $\{u_1,v_1\}$, extrapolates its calculated butterfly count by multiplying it with an appropriate increment. 
We compute the increment for each edge using the information in its corresponding cached triplet $\{s, c_b, c_g\}$. 
In specific, we use Equation~\ref{eq:reciprocal} as in %sequential 
\abacus to compute each increment and produce the partial counts, $c_0,\dots,c_{M-1}$, as shown in Figure~\ref{fig:parabacus-overview}.  
Recall that depending on whether an edge is an insertion or a deletion its partial count can be positive or negative, respectively. 
Finally, all the calculated partial counts $c_0, c_1,\dots,c_{M-1}$ are added to the old butterfly count so that the final refined count $c$ is computed. 
Note that two different versions of the sample differ slightly from each other, i.e.,~up to $M$ edges. 
Therefore, the vertex degrees among different sample versions are similar, and subsequently, the per-edge butterfly computations are balanced across threads, as we show in the experiments. 

\noindent \textbf{Version Consolidation.} The final step when processing a mini-batch is to consolidate the distinct sample versions $S_0, S_1,\dots, S_{M-1}$ into one. 
Specifically, \parabacus creates and keeps a final version of the sample that integrates all the $M$ edges in the mini-batch. 
In Figure~\ref{fig:parabacus-overview}, the final sample version also incorporates the sample update due to the edge $\{u_M,v_M\}$ into the version $S_{M-1}$, which will serve as the $0$-th version for the next mini-batch.

\subsection{Correctness and Complexity Analysis}
\label{subsec:parallel-complexity-analysis}
\noindent Here, we analyze \parabacus' correctness with respect to the butterfly counts it delivers, and its complexity with respect to both time and space when processing each mini-batch.

\begin{theorem}[Correctness]
\label{theorem:correctness-parabacus} 
\parabacus correctly counts the butterflies in a fully dynamic bipartite stream and provides the same counts as \abacus after processing each mini-batch.   
\end{theorem} 

\begin{proof}[Proof Sketch]
\label{proof:correctness-parabacus}
\parabacus first sequentially creates all the versions of the sample that an edge would observe in the order of its arrival as in \abacus. 
Subsequently, \parabacus exactly counts the per-edge butterflies formed between each edge in the mini-batch and its corresponding sample version. 
The partial counts for each edge are the same as in \abacus. 
The associativity property of the sum of the counts guarantees that the final refined butterfly count is
equal to that of %the same as in 
\abacus.  
Therefore, \parabacus achieves the same accuracy as \abacus after processing each mini-batch. 
\end{proof}

\noindent Consequently, since \parabacus provides the same butterfly count estimates as \abacus, its estimates are unbiased as well. 

\noindent \textbf{Time  Complexity.} 
We now analyze the time complexity of \parabacus in the worst case as follows:  

\begin{theorem}[Time Complexity of \parabacus]
\label{theorem:parabacus-time-complexity} 
Butterfly counting per mini-batch is performed in $O(M+\frac{M k^2}{p})$ time, where $k$ is the maximum number of edges maintained in the sample, $M$ is the number of edges in each mini-batch, and $p$ is the number of threads. 
\end{theorem} 

\begin{proof}[Proof]
\label{proof:parabacus-time-complexity}  
When processing each mini-batch, \parabacus sequentially processes all $M$ edges to create and maintain a versioned sample, which takes $O(M)$ time (O(1) for each edge). 
After constructing the versioned sample, \parabacus utilizes all $p$ available threads to conduct the per-edge butterfly counting, which takes $O(\frac{M k^2}{p})$ time. 
% Note that per-edge counting process takes $\Lambda = O(\min \{\sum_{x \in N_u} \min \{d_x,d_v\}, \sum_{x \in N_v} \min \{d_w, d_u\}\})$ time for an edge $\{u, v\}$.
Note that per-edge counting process takes $\Lambda = O(\min \{\sum_{x \in N^{\mathcal{S}}_u} \min \{d_x,d_v\}, \sum_{x \in N^{\mathcal{S}}_v} \min \{d_w, d_u\}\})$ time for an edge $\{u, v\}$.  
All the vertex degrees are upper-bounded by $k$, which is the maximum number of edges in the sample, and thus, $\Lambda = O(k^2)$. 
Therefore, the time complexity for processing each mini-batch 
is $O(M + \frac{M k^2}{p})$.  
\end{proof}

\noindent \textit{Comparison with \abacus.} Note that \abacus takes $O(M + M k^2) = O(M k^2)$ time to process a mini-batch with $M$ edges. 

\noindent \textbf{Space Complexity.} 
We provide the space analysis in Theorem~\ref{theorem:parabacus-space-complexity}. 
In a nutshell, \parabacus needs to store the $k$ elements as a sample of the input graph stream and up to $M$ elements more for the sample versions. In specific: 

\begin{theorem}[Space Complexity of \parabacus]
\label{theorem:parabacus-space-complexity} 
Butterfly counting is performed in $O(k + M)$ space, where $k$ is the number of edges maintained in the sample, and $M$ is the number of edges in each mini-batch.  
\end{theorem} 

\begin{proof}[Proof]
\label{proof:parabacus-space-complexity}  
The sample has $k$ edges. 
Since we maintain a versioned sample that stores only the deltas between versions, we have to maintain up to $M$ more edges. 
Also, we maintain a separate partial count for each edge in the mini-batch; $M$ in total. 
Therefore, \parabacus requires $O(k + M)$ space. 
\end{proof} 
\begin{table}[t!]
    \begin{center}        
        \caption{Datasets Statistics.} 
        \begin{tabular}{l|c|c|c|c|c} % |c
            % l, left-aligned
            % r, right-aligned
            % c, center-aligned
            % S, for aligning decimal numbers (for siunitx package)
            \hline
            
            \textbf{Graph} & $|E|$ & $|L|$ & $|R|$ & $B$ & \textbf{Butterfly Density} \\

            \hline
            \hline
            
            \textit{MovieLens} & 10M &  69.8K & 10.6K & 1.1T & $1.1*10^{-16}$ \\
            % \textit{Wiki-fr} & 22M &  288K &  3.9M & 601.2B & $2.5*10^{-18}$ \\
            \textit{LiveJournal} & 112M & 3.2M & 10.7M & 3.3T & $2.1*10^{-20}$ \\
            % \textit{Wiki-en} & 122M & 3.8M & 25.3M & 2,0T & $9.1*10^{-21}$ \\
            % \textit{Delicious} & 101.8M & 833K & 33.7M & 56.8B & $5.3*10^{-22}$ \\
            \textit{Trackers} & 140.6M & 27.6M & 12.7M & 20.0T & $5.1*10^{-20}$ \\
            \textit{Orkut} & 327M & 2.7M & 8.73M & 22.1T & $1.9*10^{-21}$ \\ 
            % \textcolor{blue}{\textit{MAG}} & \textcolor{blue}{1095M} & \textcolor{blue}{10.5M} & \textcolor{blue}{2.78M} & \textcolor{red}{22.1T} & \textcolor{red}{$1.9*10^{-21}$} \\
            
            % \hline 
            
            %\textit{Erd\H{o}s R\'{e}nyi}
            % \textit{Synthetic} & $100$B & $1$M & $1$M & - \\ %& - \\
            
            \hline
        \end{tabular} 
        \label{tab:datasets}
  \end{center}
  % \hspace{-2cm}
\end{table}

\begin{figure*}[t]
  \begin{subfigure}[t]{0.24\textwidth}  
    \includegraphics[width=\textwidth]{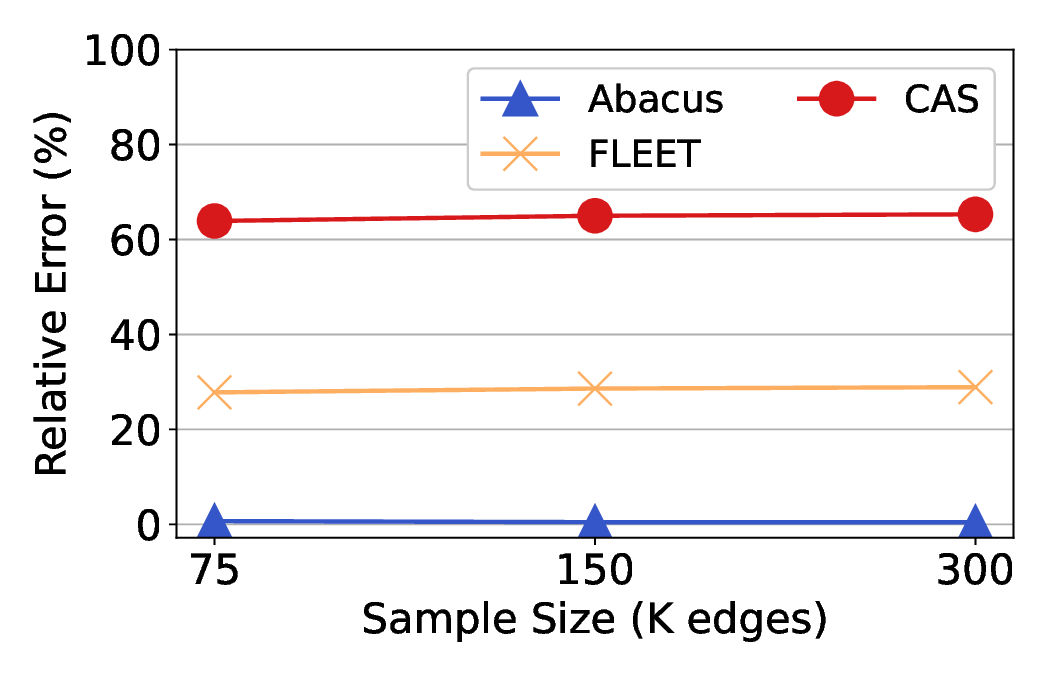} 
    \caption{\textit{Movielens}}
    \label{fig:acc-varying-membudget-movielens}
  \end{subfigure}
  \begin{subfigure}[t]{0.24\textwidth}  
    \includegraphics[width=\textwidth]{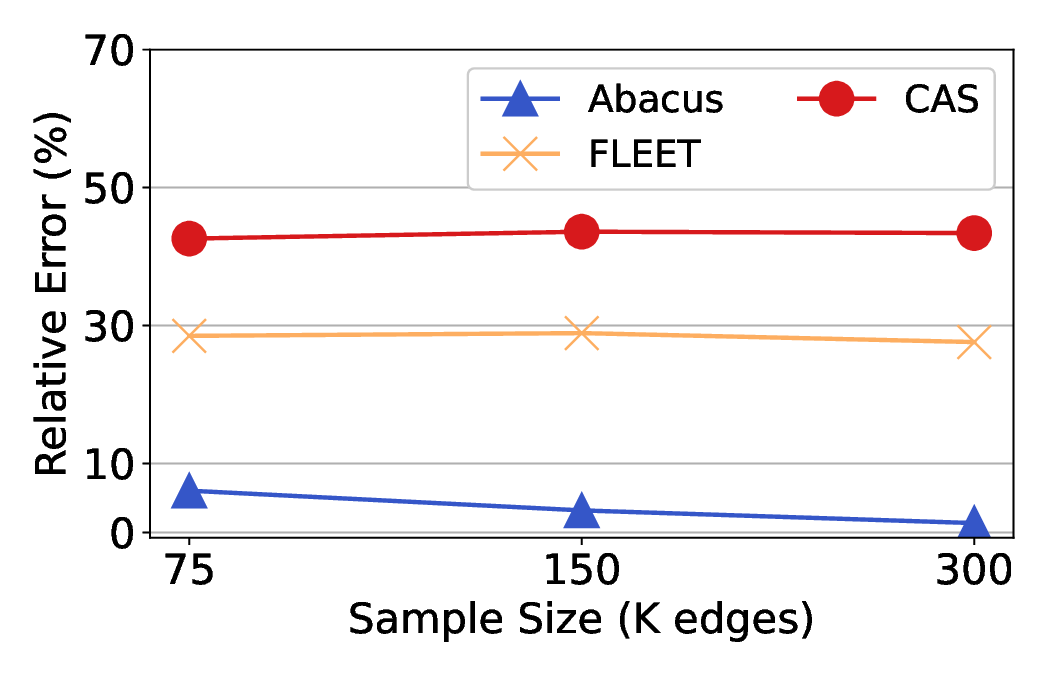}
    \caption{\textit{LiveJournal}}
    \label{fig:acc-varying-membudget-livejournal}
  \end{subfigure}
  \begin{subfigure}[t]{0.24\textwidth}  
    \includegraphics[width=\textwidth]{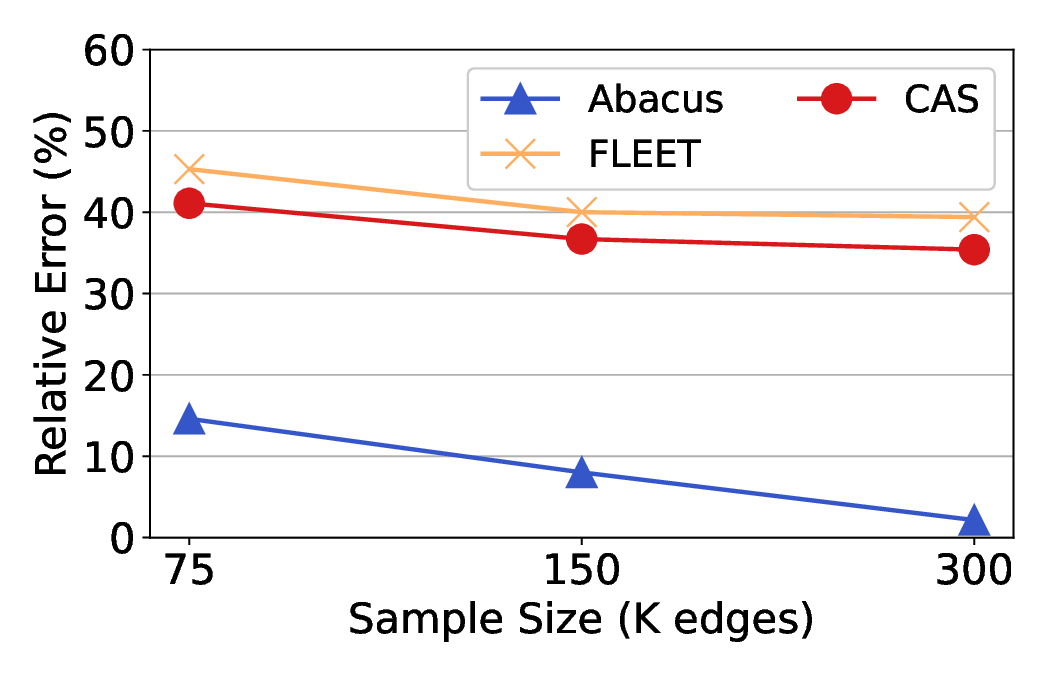}
    \caption{\textit{Trackers}} 
    \label{fig:acc-varying-membudget-trackers}
  \end{subfigure} 
  \begin{subfigure}[t]{0.24\textwidth}  
    \includegraphics[width=\textwidth]{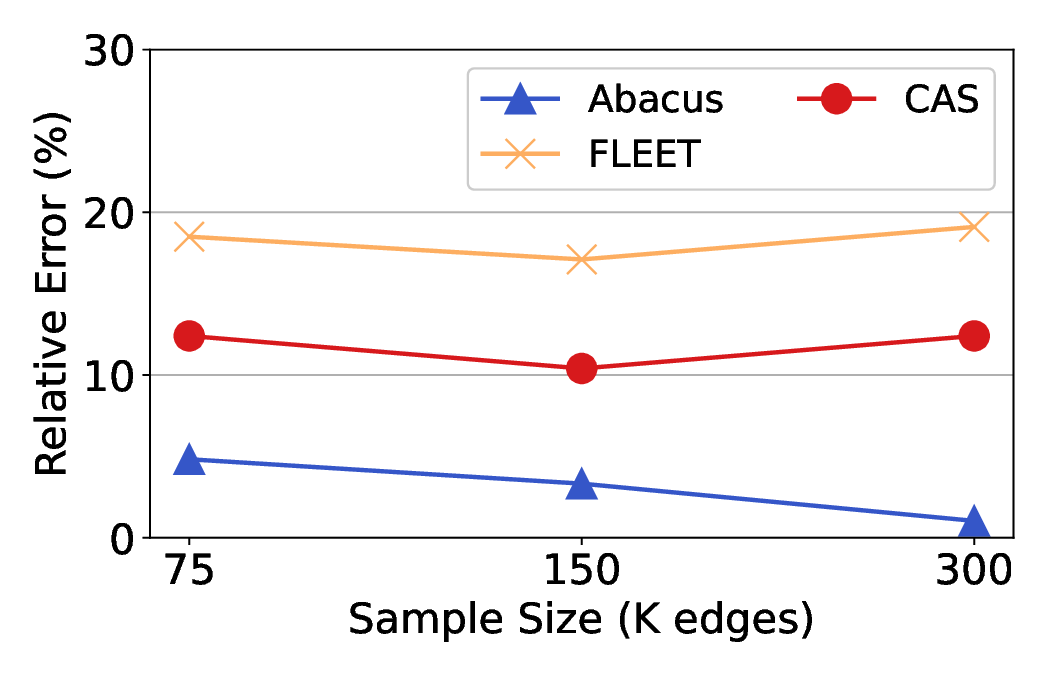}
    \caption{\textit{Orkut}} 
    \label{fig:acc-varying-membudget-orkut}
  \end{subfigure} 
  \caption{Relative Error of \abacus with $20\%$ of deletions while varying the sample size %memory budget 
  of edges.  Discarding deletions, as in FLEET and CAS, negatively impacts the accuracy.
  }
  \label{fig:accuracy-varying-memory-budget}
\end{figure*}

%%%%%%%%%%%%%%%%%%%%%%%%%%%%%%%%%%%%%%%%%%%%%
%%%%%%%%%%%%%%%%%%%%%%%%%%%%%%%%%%%%%%%%%%%%%
\section{Experimental Evaluation}
\label{sec:experiments}
\noindent We evaluate \abacus/\parabacus using %a variety of
four large-scale real-world %and synthetic 
bipartite graphs and investigate: how effective it is in terms of the error in butterfly estimation; how efficient it is in terms of throughput; %and total running time; 
how it is affected by the amount of edge deletions;
how it scales to large %bipartite 
graph streams; 
and, how much the speedup of its parallel version, \parabacus, is affected by the mini-batch size and number of threads. 
% when using multi-threaded and load balanced execution; 
%; and whether it can spot anomalies in adversarial graph streams.

\begin{myboxi}
Overall, our major findings include that \abacus/\parabacus: 
(i)~achieves significantly higher (up to $148\times$ better) accuracy than the baselines, which inherently cannot handle deletions majorly affecting accuracy,  
(ii)~it has similar throughput with its competitors when processing edge insertions using one thread and much higher throughput in its parallel version when using multiple threads, 
(iii)~it is consistently accurate irrespective of the ratio of deleted edges, 
(iv)~it scales linearly to the number of edges in a graph, 
(v)~\parabacus  accomplishes considerable speedup that depends on the graph characteristics. 
Finally, our solution delivers unbiased estimates as proved in Section~\ref{sec:accuracy-and-complexity}. 
\end{myboxi}

\subsection{Experimental Setup}
\label{subsec:experimental-setup}
\noindent \textbf{Hardware.} We ran our experiments on a server with a 10-core Intel(R) Xeon(R) Gold 5115 CPU @ 2.40GHz with 4-way hyper-threading and $188$GB of main memory.

\noindent \textbf{Implementation.} We implemented \abacus and our baselines in Java. 
For the parallel version of \abacus we used the Callable Java Interface. 
Note that we store the sampled edges using the adjacency list format.

\noindent \textbf{Datasets.} We used four real-world bipartite graphs from the Koblenz Network Collection~\footnote{\url{http://konect.uni-koblenz.de/}}  (KONECT)~\cite{konect}, whose characteristics we show in 
Table~\ref{tab:datasets}. 
% \textcolor{blue}{\textit{Real Graphs} -- } 
\textit{\textbf{MovieLens}} contains movie ratings by users. 
\textit{\textbf{LiveJournal}} is a bipartite graph of the LiveJournal social network with users and their group memberships.  
\textit{\textbf{Trackers}} is a bipartite graph
of internet domains and trackers, where each edge represents that a tracker is identified by its domains. 
\textit{\textbf{Orkut}} consists of group-user relationships where edges represent the group memberships of users. 
% \textcolor{blue}{\textit{\textbf{MAG}}\footnote{\url{https://figshare.com/articles/dataset/mag_scholar/12696653}}~\cite{sigmod22-yang} is a bipartite graph containing a set of papers and a set of words extracted from abstracts as nodes, and a set of edges connecting papers and words with word occurrence as edge weights.}  
We preprocessed and converted the 
%\textcolor{blue}{real} 
graphs to be undirected and unweighted. 
Also, we removed duplicate edges, self-loops, and zero-degree vertices. 
In all our datasets, we simulate the stream assuming that an edge arrives at each discrete time $t \geq 1$. 
All edges arrive in their natural order as in the datasets. 
% \textcolor{blue}{\textit{Synthetic Graphs} -- We generated large-scale synthetic graphs sampled from the R-MAT model. 
% Specifically, we used the TrillionG\footnote{https://github.com/chan150/TrillionG} \cite{sigmod17-trillionG} tool to generate \textit{\textbf{Erd\H{o}s R\'{e}nyi, er-$k$,}} graphs  with $2^k$ nodes, uniformly distributed edges and with an average vertex degree of $100$ by setting the R-MAT parameters to $a = b = c = d = 0.25$.  
% % Additionally, we varied $k$ from $16$ to $22$ to evaluate the scalability of \abacus. 
% We also generated a set of skewed graphs, \textit{\textbf{sg-$s$}}, with $2^{20}$ nodes with an average degree of $10$, while varying the skew. 
% We set the R-MAT parameters ($a$, $b$, $c$, $d$) so that the number of edges in the bottom-right part of the matrix is about $s$ times the top-left part of the matrix. 
% We set $b = c = 0.25$. 
% Thus, when $s = 1$, there is no skew, while if $s > 1$,~R-MAT generates power-law graphs.  
% We varied $s$ from $1$ to $7$ with a step of $2$. 
% } 
\textit{Deletions. }
Our real 
datasets are insertion-only by default, and thus, we generate fully dynamic graph streams % consisting of both insertions and deletions
by %actually 
generating the edge deletions from the %bipartite 
graphs listed in Table~\ref{tab:datasets}. 
%We generate fully dynamic graph streams consisting of both insertions and deletions from the bipartite graphs listed in Table~\ref{tab:datasets}. 
In specific, we (a) create the insertions of each edge in the input bipartite graphs using their natural order, (b) create the deletions by selecting $\alpha\%$ of the edges from the input bipartite graphs, (c) place each created deletion in a random position after its corresponding insertion. 
%As discovered in the thorough analysis of~\cite{DBLP:conf/cscw/AlmuhimediWLSA13}, there are up to $30\%$ edge deletions in real-world Twitter graphs. 
We use $\alpha=20\%$ as our default value and 
% experiment with different values of $\alpha$ 
assess the impact of varying $\alpha$ in Section~\ref{subsec:impact-of-deletions-a}.
These values stem from~\cite{DBLP:conf/cscw/AlmuhimediWLSA13} which reports up to $30\%$ edge deletions in real-world Twitter graphs.  

\begin{figure*}[t]
  \begin{subfigure}[t]{0.24\textwidth}  
    \includegraphics[width=\textwidth]{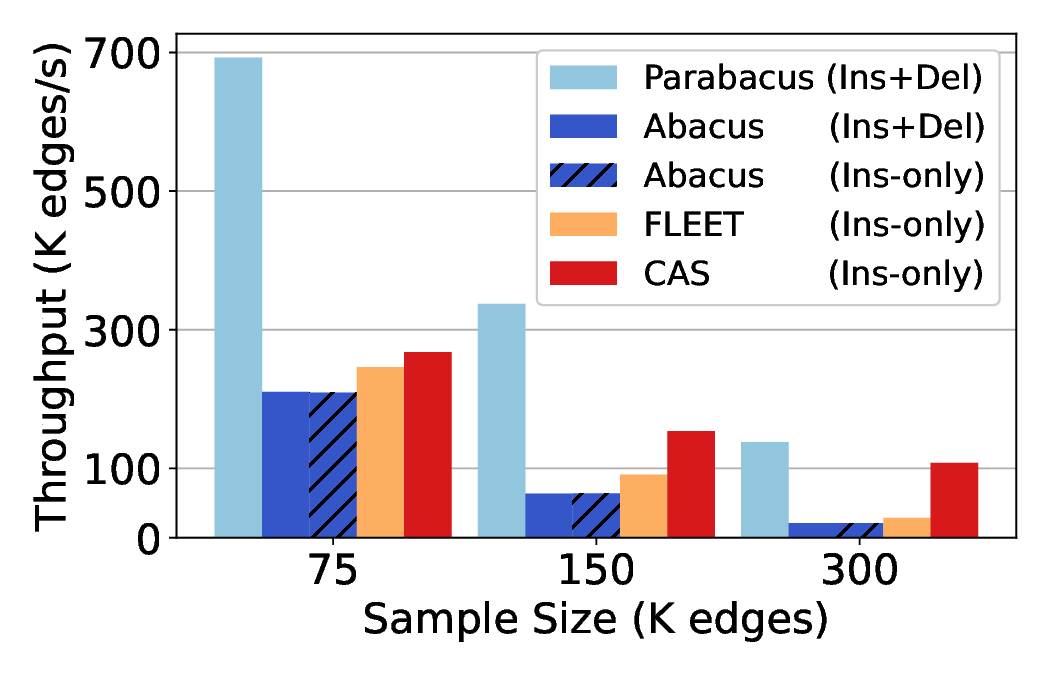}
    \caption{\textit{Movielens}}
    \label{fig:throughput-varying-alpha-movielens}
  \end{subfigure}
  \begin{subfigure}[t]{0.24\textwidth}   
    \includegraphics[width=\textwidth]{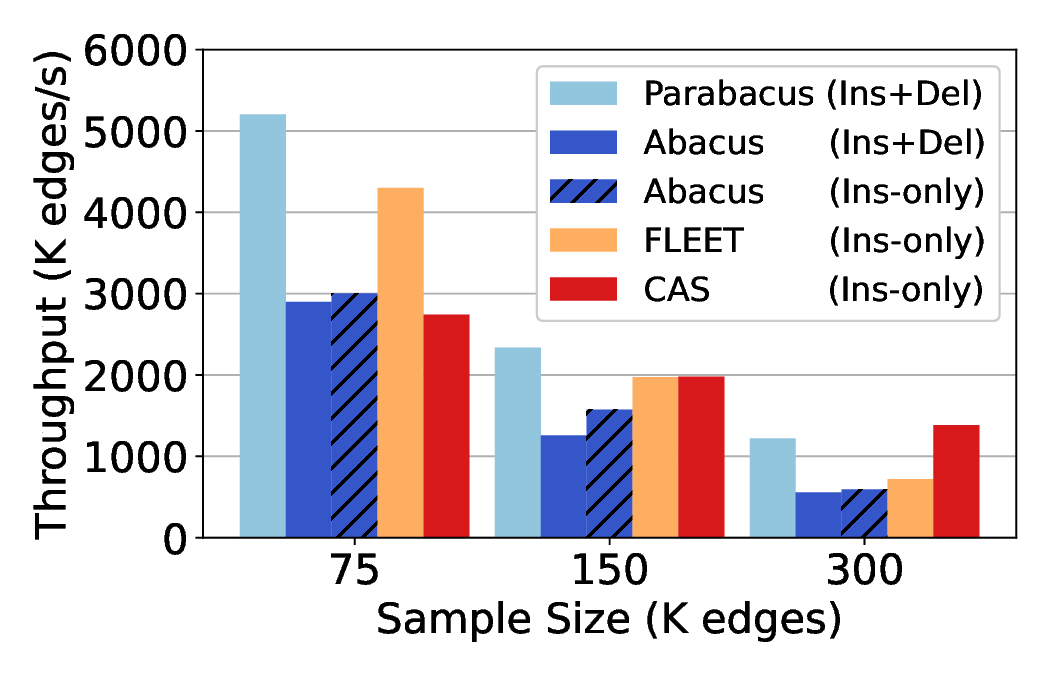}
    \caption{\textit{LiveJournal}}
    \label{fig:throughput-varying-alpha-livejournal}
  \end{subfigure}
  \begin{subfigure}[t]{0.24\textwidth}   
    \includegraphics[width=\textwidth]{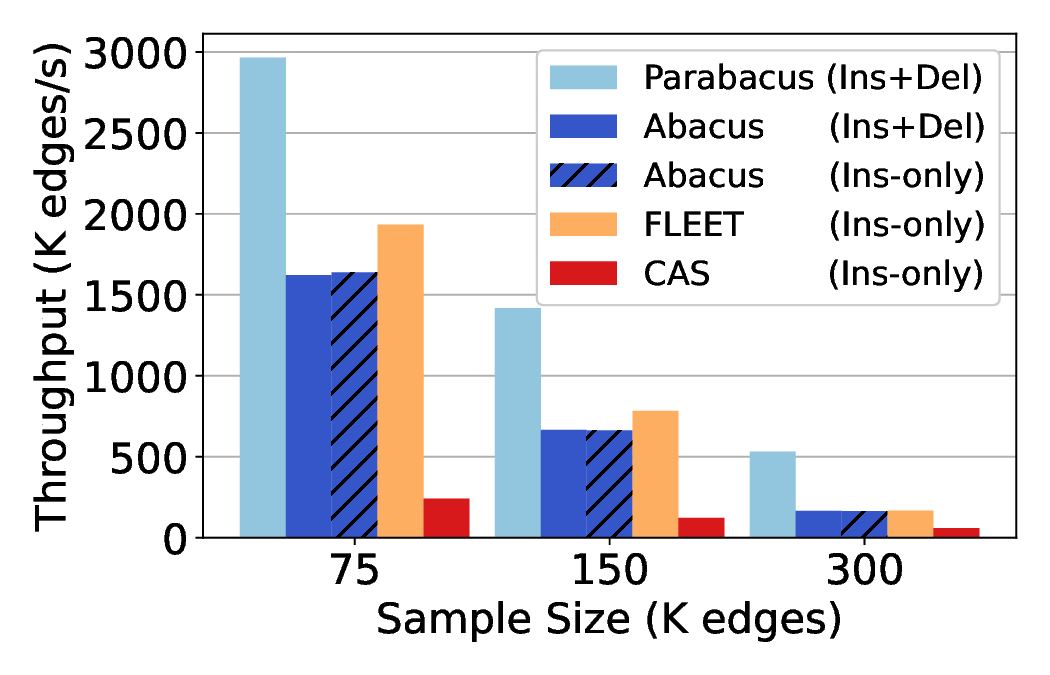}
    \caption{\textit{Trackers}} 
    \label{fig:throughput-varying-alpha-trackers}
  \end{subfigure} 
  \begin{subfigure}[t]{0.24\textwidth}  
     \includegraphics[width=\textwidth]{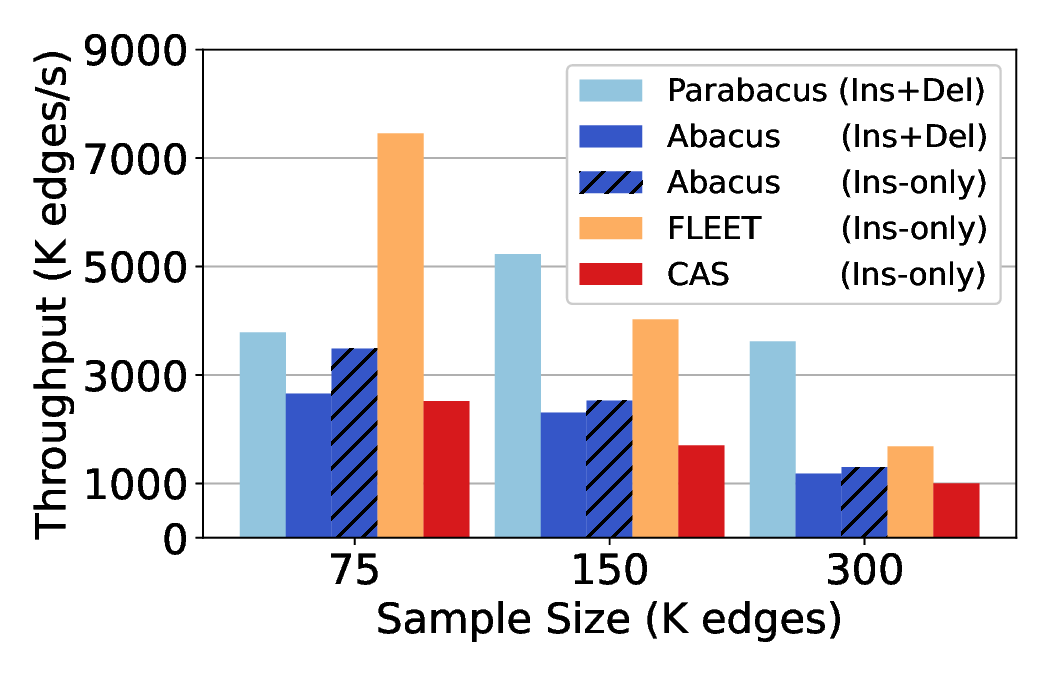}
    \caption{\textit{Orkut}} 
    \label{fig:throughput-varying-alpha-orkut}
  \end{subfigure} 
  \caption{Throughput   for all datasets with $20\%$ of deletions, while varying the sample size of edges. Notably \abacus achieves a similar throughput with the baselines when processing either insertions only (Ins-only) or both insertions and deletions (Ins+Del). 
  \parabacus achieves a significantly higher throughput for a small mini-batch size of $500$ edges. 
  }
  \label{fig:throughput-varying-membudget} 
\end{figure*}

\noindent \textbf{Baselines.} To our knowledge, no existing solution considers estimating butterfly counts on fully dynamic graph streams, entailing both insertions and deletions. 
Yet, we compare \abacus/\parabacus with FLEET~\cite{cikm19-fleet} and CAS~\cite{tkde21-cas}, which are the state-of-the-art approaches for insertion-only bipartite graph streams and are the most relevant techniques to our problem.
%	We included FLEET and CAS as baselines because they are the most relevant techniqeus to our problem, yet, they do not handle edge deletions and they discard each edge deletion without updating their butterfly count estimations.
We do so, first, to quantify the effect of disregarding edge deletions on accuracy, and, second, to compare the throughput of \abacus and its parallel variant, \parabacus, with that of the best available solutions designed for insertion-only streams. 
In specific, we use FLEET3, the best method of~\cite{cikm19-fleet}, with a reservoir resizing parameter $\gamma = 0.75$ as proposed, and CAS-R, the best method of~\cite{tkde21-cas}, with the ratio of memory usage of AMS sketch to total memory equal to $\lambda = 0.33$ as proposed. 
For \parabacus, we use a mini-batch size of $500$ and $40$ threads unless indicated otherwise.

\noindent \textbf{Evaluation Metrics.} Let $x$ be the true butterfly count and let $\hat{x}$ be the corresponding estimate obtained by the evaluated algorithm. 
For evaluating the accuracy of a method, we use the \textit{relative error} metric (the lower the better), which is defined as $\frac{|x - \hat{x}|}{x}$, for a true butterfly count $x$ that is greater than zero.

\subsection{Accuracy}
\label{subsec:accuracy-global}
%%%
\noindent We first investigate the accuracy of estimating butterfly counts in the presence of edge deletions. 
\abacus and \parabacus demonstrate the same accuracy and, thus, we denote our solution as \abacus for simplicity.

We vary the sample size from $75K$ to $300K$ edges. 
We run each experiment $10$ times and show the average relative error values in Figure~\ref{fig:accuracy-varying-memory-budget}. 
\abacus provides $40.6-65.7\times$ more accurate counts than FLEET in \textit{Movielens}, $4.7-20.3\times$ in \textit{Livejournal}, $3.8-18.5\times$ in \textit{Trackers}, and $3.2-18.4\times$ in \textit{Orkut}. 
Also, \abacus provides $93.4-148.4\times$ more accurate counts than CAS in \textit{Movielens}, $7.1-31.9\times$ in \textit{Livejournal}, $2.81-16.54\times$ in \textit{Trackers}, and $2.57-12.03\times$ in \textit{Orkut}.

% FLEET and CAS do not handle edge deletions and they discard each edge deletion without updating their butterfly count estimations.  
This clearly indicates that edge deletions have a significant impact on butterfly counts if they are ignored, and therefore, it is essential to handle them carefully in order to achieve optimal performance. 
Furthermore, \abacus is capable of maintaining quite accurate counts on all datasets and achieves on average $~0.52\%$ relative error on \textit{Movielens}, $~3.54\%$ on \textit{Livejournal}, $~8.25\%$ on \textit{Trackers}, and $~3.05\%$ on \textit{Orkut}. 
Additionally, we observe that the relative error decreases as the sample size increases. 
For instance, in \textit{Livejournal} \abacus achieves $6.06\%$ error for a sample size of $75K$ edges, and only $1.36\%$ error when maintaining $300K$ edges. 
We observe a similar pattern across the remaining datasets.  
As the sample size grows, \abacus can more precisely calculate the number of butterflies that each incoming edge forms with the edges in the sample, allowing for a more accurate estimation. 
This is not the case for the baselines, because they discard the edge deletions. 
Consequently, their samples are not representative of the fully dynamic graph streams irrespective of the sample size. 
\textit{Therefore, we conclude that (i) it is imperative to account for edge deletions when estimating butterfly counts, and (ii) \abacus accurately estimates the butterfly counts in bipartite graph streams containing both edge insertions and deletions.}

\begin{figure*}[t]
  \begin{subfigure}[t]{0.24\textwidth}  
    \includegraphics[width=\textwidth]{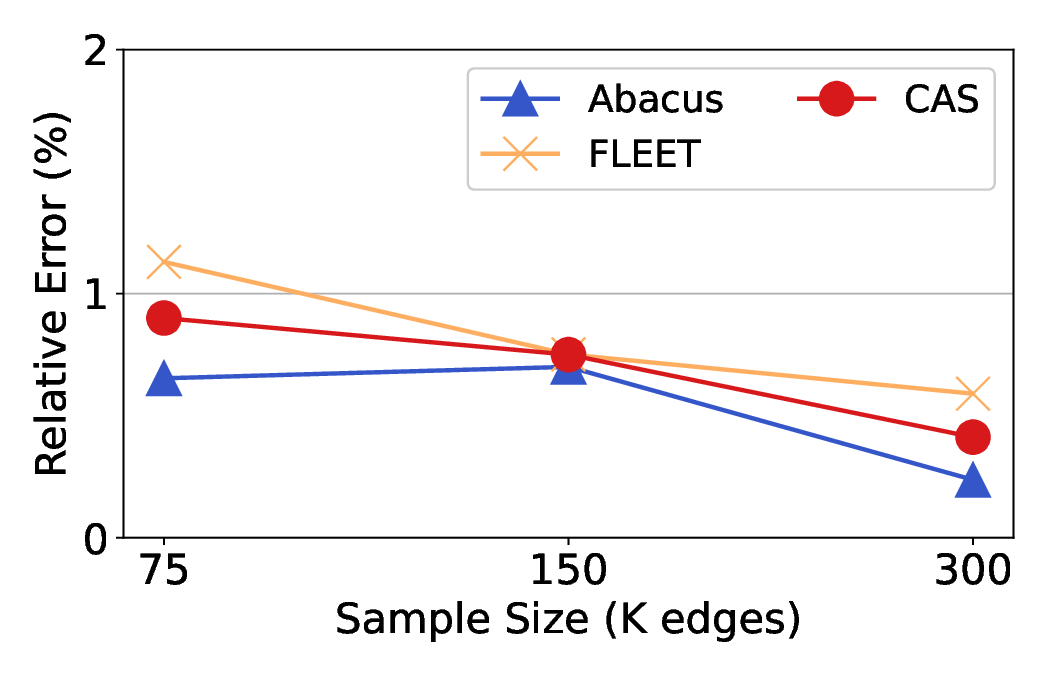} 
    \caption{\textit{Movielens}}
    \label{fig:acc-insertion-only-movielens}
  \end{subfigure}
  \begin{subfigure}[t]{0.24\textwidth}  
    \includegraphics[width=\textwidth]{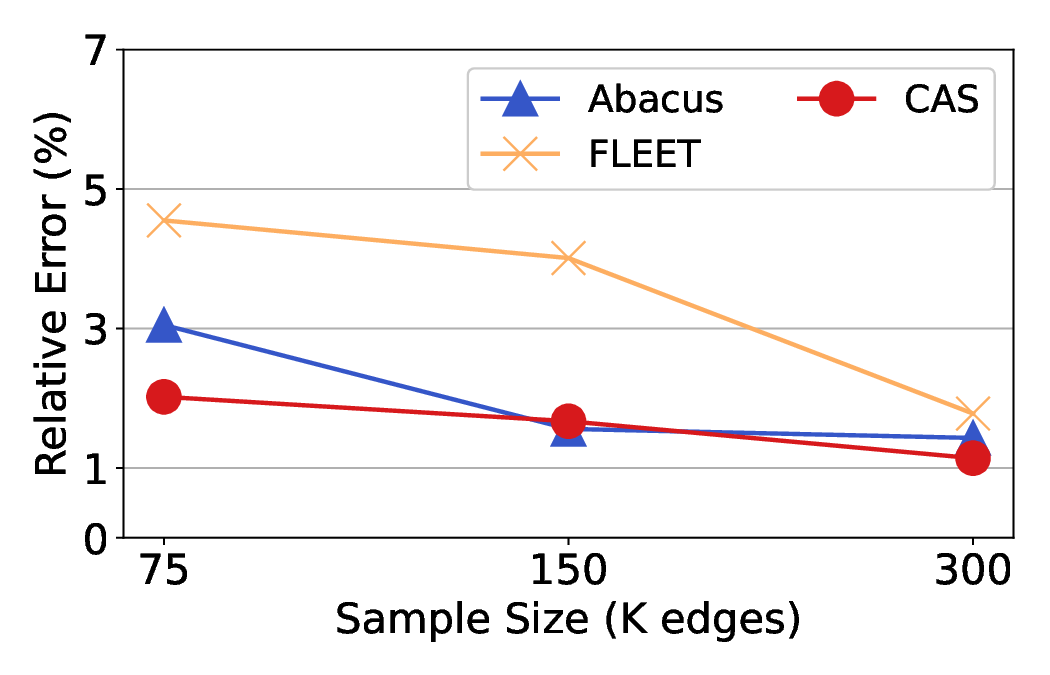}
    \caption{\textit{LiveJournal}}
    \label{fig:acc-insertion-only-livejournal}
  \end{subfigure}
  \begin{subfigure}[t]{0.24\textwidth}  
    \includegraphics[width=\textwidth]{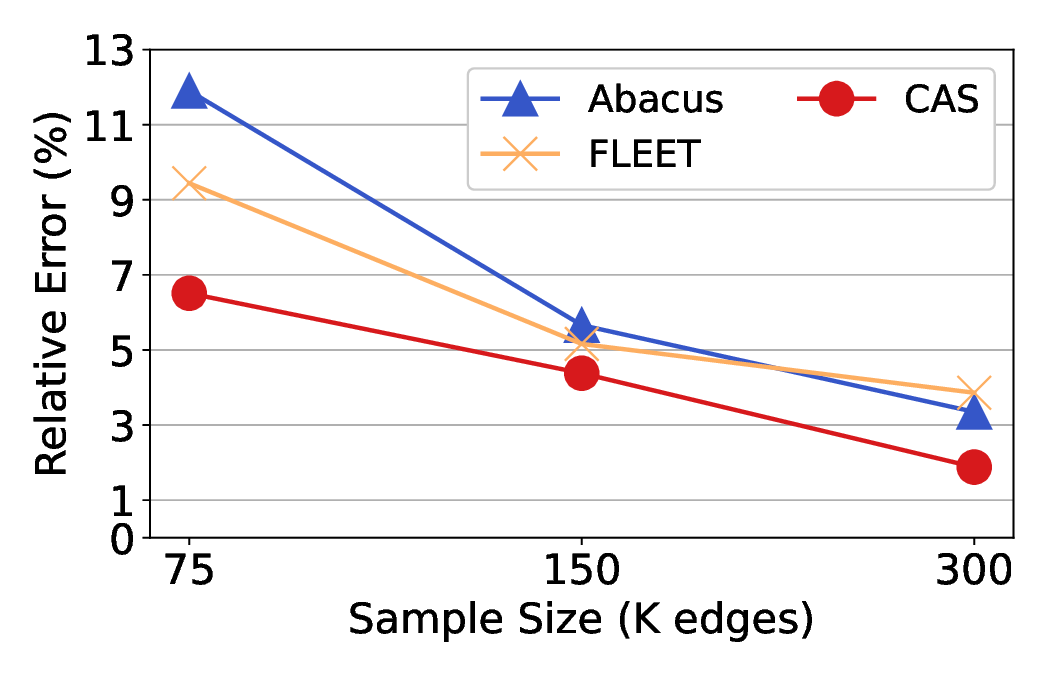}
    \caption{\textit{Trackers}} 
    % \makis{rerun this}
    \label{fig:acc-insertion-only-trackers}
  \end{subfigure} 
  \begin{subfigure}[t]{0.24\textwidth}  
    \includegraphics[width=\textwidth]{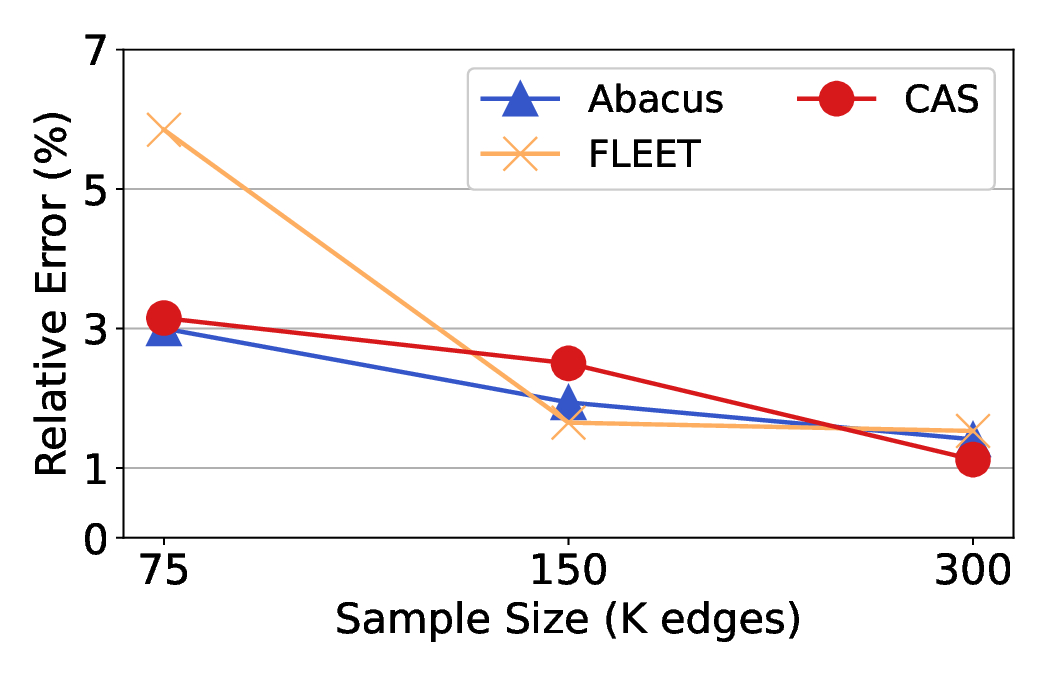}
    \caption{\textit{Orkut}} 
    \label{fig:acc-insertion-only-orkut}
  \end{subfigure} 
  \caption{Relative Error of \abacus on insertion-only streams ($\alpha = 0\%$), while varying the sample size %memory budget 
  $k$ of edges maintained.  
  }
  \label{fig:accuracy-insertion-only}
\end{figure*}

\subsection{Throughput}
\label{subsec:throughput}
%%%

\noindent We now compare \abacus/\parabacus with FLEET and CAS in terms of throughput, i.e., the number of edges processed per second. 
For both FLEET and CAS, we set the reservoir size equal to the sample size in \abacus/\parabacus.  
For calculating the throughput, we measure the running time of each method independently of the ingestion rate of the input graph stream ignoring the waiting time for each edge's arrival.

Figure~\ref{fig:throughput-varying-membudget} illustrates the throughput that \abacus, \parabacus, %and 
FLEET, and CAS achieve when processing input graph streams with insertions and deletions ($\alpha = 20\%$).  
For \abacus/\parabacus, we show the throughput it achieves for processing both insertions and deletions. 
For a fair comparison with the baselines that do not support deletions, we also show the throughput of \abacus for processing the insertions only (Ins-only). 
In general, we observe that \abacus achieves a throughput close to that of FLEET and CAS, not only in the case of insertions-only but also in the case where \abacus handles deletions. 
In addition, we see that \parabacus significantly enhances the throughput and by far surpasses the baselines in the majority of cases, even with a relatively small mini-batch size of $500$ edges. 
In specific, \parabacus achieves up to $4.85\times$ higher throughput than FLEET, and up to $12.26\times$ higher throughput than CAS, without sacrificing the accuracy.
The throughput enhancement increases when using a larger mini-batch size, as we show later in Section~\ref{subsec:parallelization}.

In more detail, we observe that \abacus' throughput when processing insertions is similar to the throughput of FLEET for sample sizes of $150K$ and $300K$ edges.  
However, for smaller sample sizes such as $75K$ edges, FLEET attains approximately up to $~1.5\times$ higher throughput than \abacus. 
This happens because FLEET always maintains a non-full sample as it resizes its sample and keeps only the $75\%$ of it every time it reaches its maximum capacity. 
Therefore, the per-edge butterfly counting after each edge's arrival is conducted using a consistently smaller sample than in \abacus.
Another reason for this corner case is the low density of \textit{Orkut}, which leads to even less number of butterlies to be formed in the maintained sample. 
In addition, we see that \abacus achieves a similar throughput to CAS, except in \textit{Trackers} graph where CAS attains a lower throughput. 
We found out that around half of the time in CAS is attributed to the update of the sketch it reserves. 
Finally, for all approaches, we see that, in general, the more edges their sample has the lower the achieved throughput.
This is reasonable since more work related to butterfly counting is required for processing the same bipartite graph stream.  

\noindent \textit{We conclude that \abacus achieves a throughput very close to that of FLEET and CAS and \parabacus can achieve an order of magnitude higher throughput even when using a small mini-batch size of $500$ edges. 
Consequently, performance is not sacrificed when processing edge deletions.}

\begin{figure}[t]
  \begin{subfigure}[t]{0.235\textwidth}  
    \includegraphics[width=\textwidth]{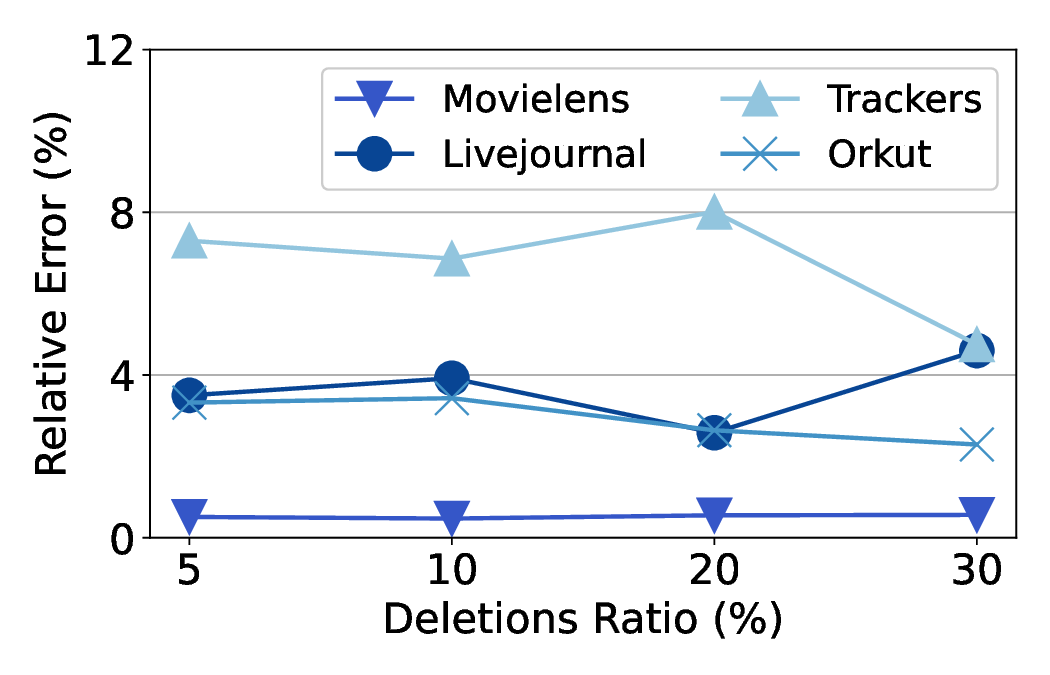}
    \caption{Relative Error}  
    \label{fig:alpha-deletions-accuracy-relative-error}
  \end{subfigure}
    \begin{subfigure}[t]{0.235\textwidth}  
    \includegraphics[width=\textwidth]{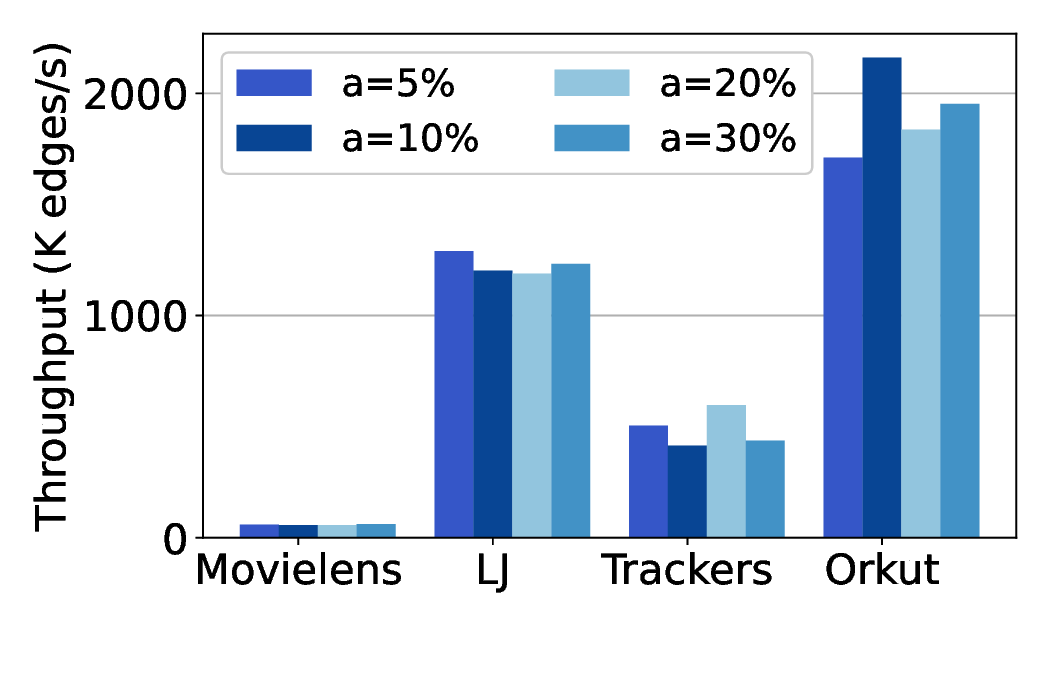}
    \caption{Throughput}  
    \label{fig:alpha-deletions-throughput}
  \end{subfigure} 
  \caption{Impact of deletions $\alpha$ on %\abacus's 
  accuracy and throughput.} 
  \label{fig:impact-deletions-a} 
\end{figure}

\subsection{Accuracy for Insertion-only Streams}
\label{subsec:insertion-only-streams}
%%%
\noindent We now compare \abacus with FLEET and CAS in terms of accuracy (i.e., relative error) when processing insertion-only bipartite graph streams, i.e.,~$\alpha=0\%$. 

Figure~\ref{fig:accuracy-insertion-only} shows the accuracy in terms of relative error that \abacus, FLEET, and CAS achieve over bipartite graph streams that contain no deletions.  
We vary the sample size from $75K$ to $300K$ edges.  
We run each experiment 10 times and show the average relative error values in Figure~\ref{fig:accuracy-insertion-only}. 
We observe that \abacus maintains accuracy comparable to that of FLEET, and is even more accurate in \textit{Movielens} and \textit{Livejournal} bipartite graph streams.  
We attribute this to the fact that \abacus maintains a sample that has a maximum size equal to its memory budget and always strives to keep its sample full, whereas FLEET resizes its sample every time it becomes full and keeps only $75\%$ of the edges it contains.  
Furthermore, \abacus achieves similar accuracy to CAS indicating that it does not exhibit deficiencies in the absence of deletions. 
In addition, we observe that the relative error decreases as the sample size increases. 
For instance, \abacus achieves $11.9\%$ relative error in \textit{Trackers} for a sample size of $75K$ edges, and only $3.35\%$ error when maintaining $300K$ edges in its sample. 
We observe a similar trend in the rest of the datasets. 
This holds for \abacus, FLEET, and CAS because the more edges stored in their sample, the more precisely they 
estimate the butterfly counts. 
\textit{We conclude that \abacus provides butterfly count estimations on insertion-only streams that are at least as accurate as the methods designed specifically for processing insertion-only streams.}

\subsection{Impact of Deletions}
\label{subsec:impact-of-deletions-a}
\noindent We proceed in exploring the impact of deletions ratio, $\alpha$, on the accuracy (i.e., relative error) and throughput (i.e., number of edges processed per second) of our approach. 
We use a sample size of $150K$ edges and vary $\alpha$ from $5\%$ up to $30\%$.

Figure~\ref{fig:alpha-deletions-accuracy-relative-error} illustrates the relative error for the butterfly count that \abacus entails when varying the actual ratio of deletions $\alpha$. 
We observe that \abacus produces relatively accurate butterfly counts by maintaining a sample of only $150K$ edges irrespective of the size of the graph stream. 
Specifically, the relative error in all of our datasets is less than $8\%$. 
Furthermore, we observe that the relative error of \abacus is consistent across datasets and is unaffected by the ratio $\alpha$, indicating that \abacus maintains a small error in the presence of deletions regardless of their number. 
In addition, Figure~\ref{fig:alpha-deletions-throughput} shows the effect of $\alpha$ on the overall throughput of \abacus when processing an input bipartite stream with deletions. 
Note that the bigger the deletions ratio $\alpha$ the more edges exist in a graph stream in total. 
Despite this, \abacus maintains a constant throughput for a given dataset, regardless of the deletions ratio $\alpha$. 
Also, note that the throughput of \abacus varies based on the dataset. 
The graph characteristics, such as the butterfly density of a  graph affect the number of butterflies observed after the arrival of each edge and, thus, lead to different throughput for each distinct graph stream. 
\textit{Therefore, we conclude that \abacus provides consistently accurate butterfly count estimations and maintains steady throughput regardless of the ratio of deletions in the bipartite graph stream.}

\begin{figure}[t]
  \begin{subfigure}[t]{0.235\textwidth}   
    \includegraphics[width=\textwidth]{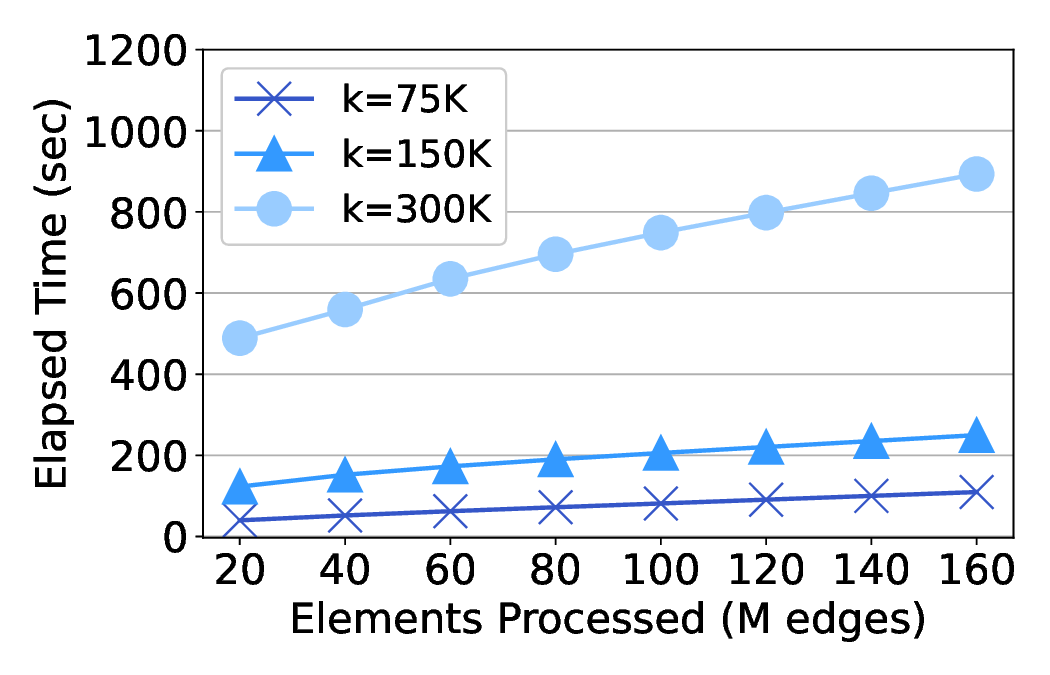}
    \caption{\textit{Trackers}}  
    \label{fig:scalability-trackers}
  \end{subfigure}
    \begin{subfigure}[t]{0.235\textwidth}   
    \includegraphics[width=\textwidth]{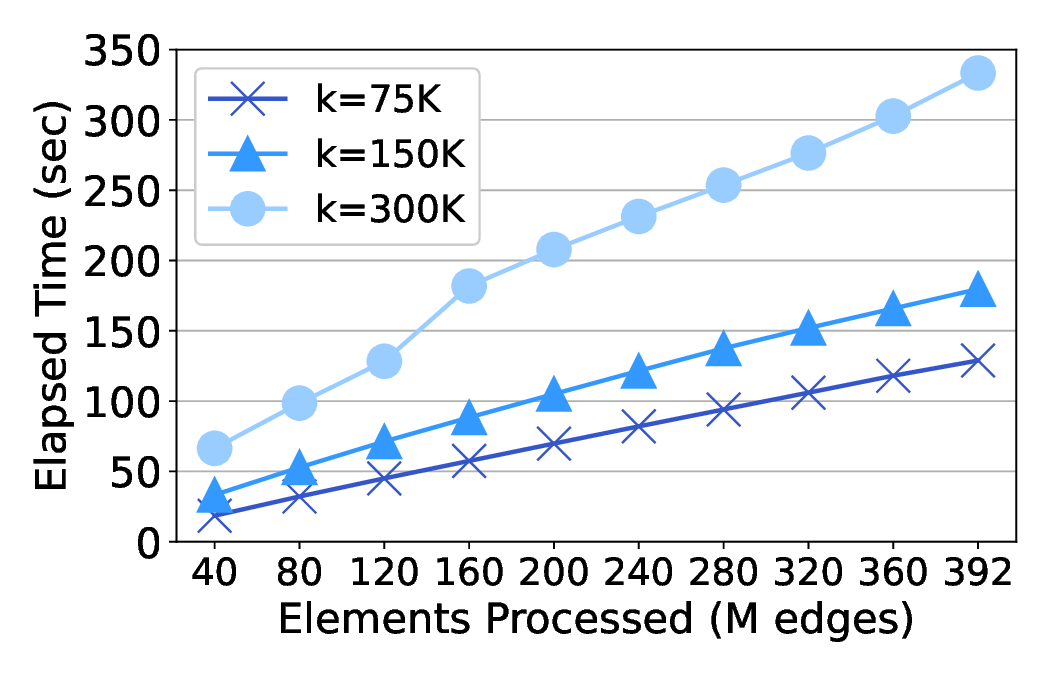}
    \caption{\textit{Orkut}}  
    \label{fig:scalability-orkut}
  \end{subfigure} 
  \caption{\abacus scales linearly with the input stream size.} 
  \label{fig:scalability-input-size} 
\end{figure}

\begin{figure*}[hbtp]
  \centering
  \begin{subfigure}[t]{0.24\textwidth}
    \includegraphics[width=\textwidth]{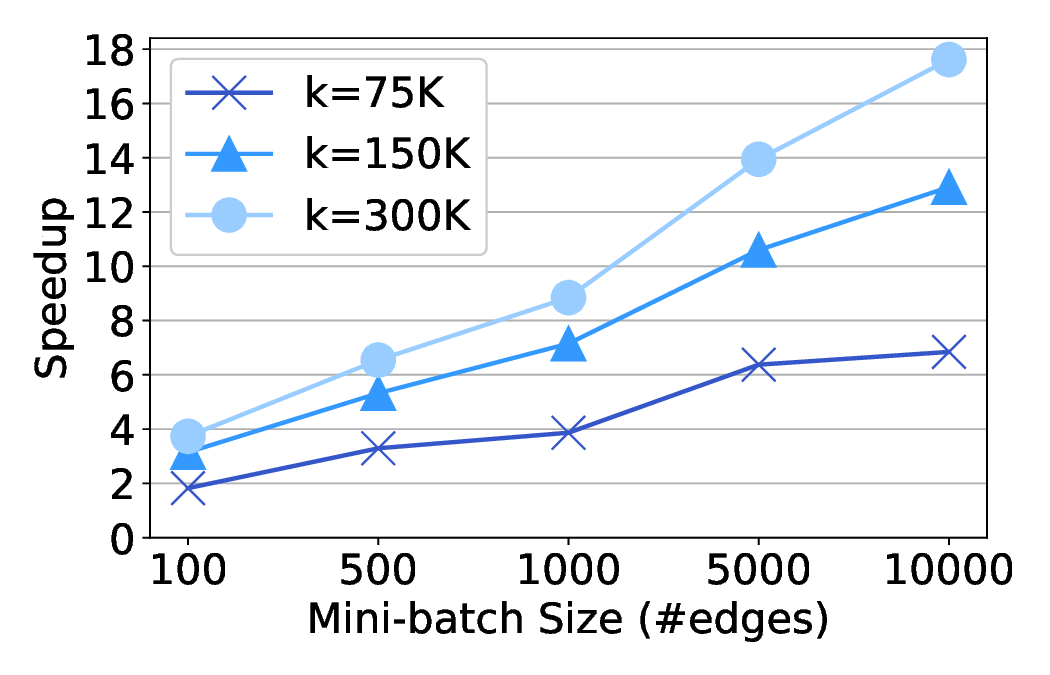} 
    \caption{\textit{MovieLens}}
    \label{fig:parallelization-ex2-movielens}
  \end{subfigure}
  \begin{subfigure}[t]{0.24\textwidth} 
    \includegraphics[width=\textwidth]{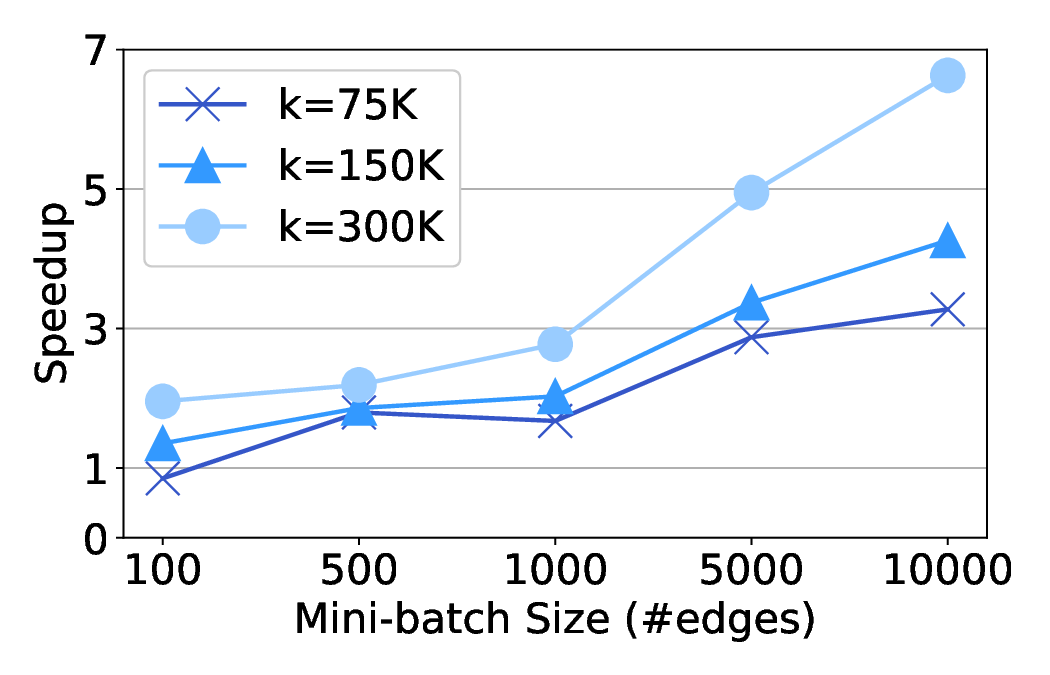}
    \caption{\textit{LiveJournal}}
    \label{fig:parallelization-ex2-livejournal}
  \end{subfigure}
  \begin{subfigure}[t]{0.24\textwidth} 
    \includegraphics[width=\textwidth]{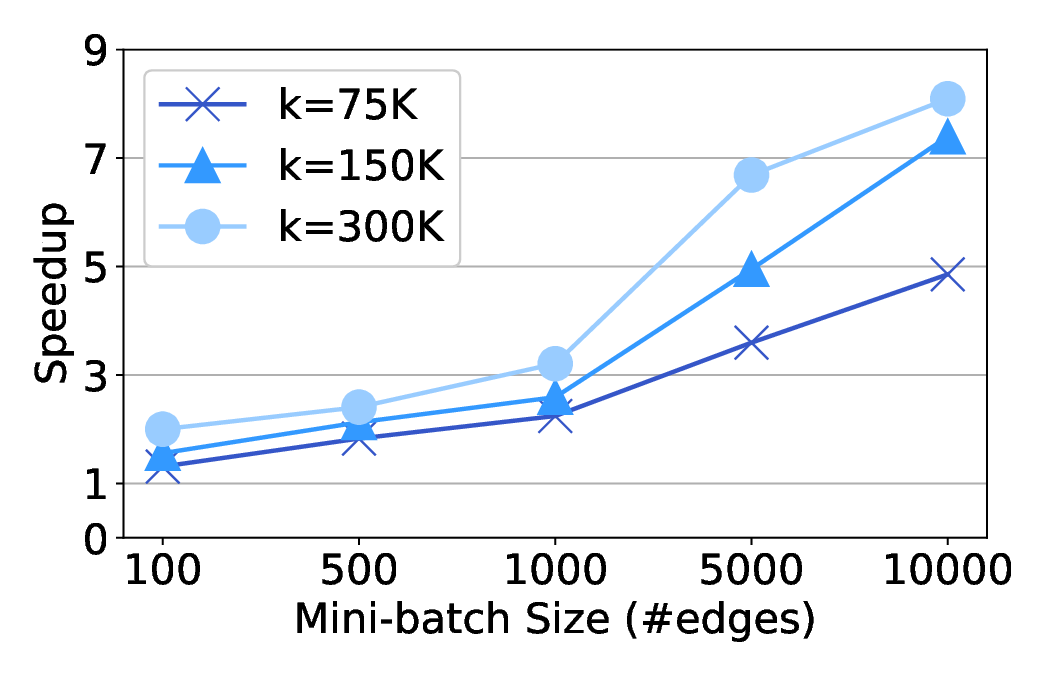}
    \caption{\textit{Trackers}}
    \label{fig:parallelization-ex2-trackers}
  \end{subfigure}
  \begin{subfigure}[t]{0.24\textwidth} 
    \includegraphics[width=\textwidth]{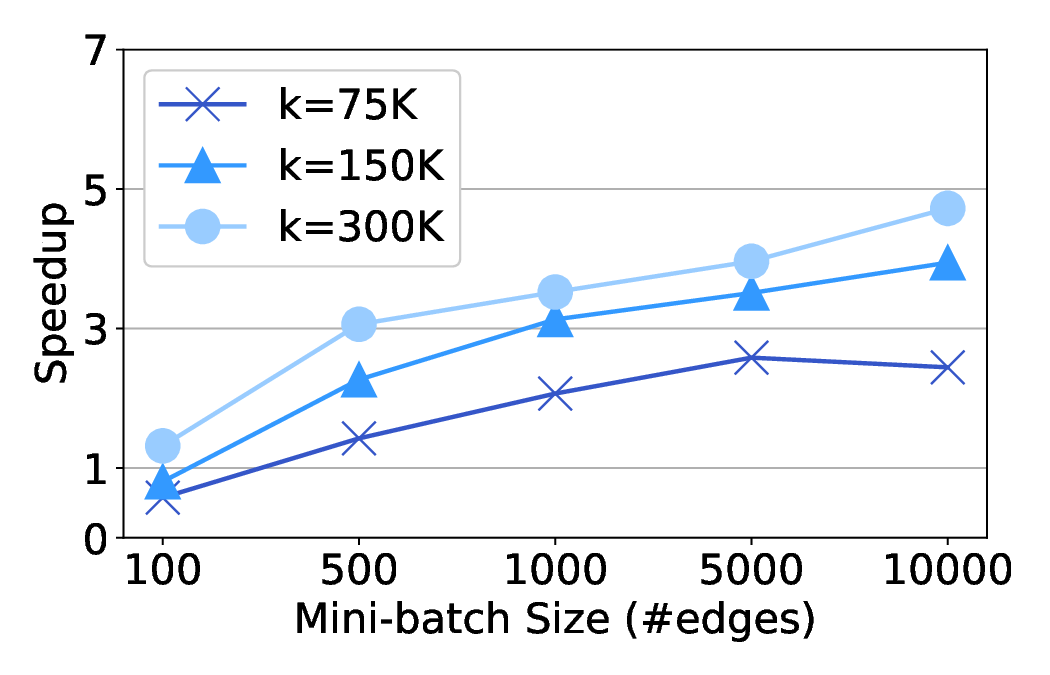}
    \caption{\textit{Orkut}}
    \label{fig:parallelization-ex2-orkut}
  \end{subfigure} 
  \caption{Speedup of \parabacus when varying the mini-batch size and using all $40$ threads and with a fixed sample size.}
  \label{fig:parallelization-minibatch}
\end{figure*}

\subsection{Scalability}
\label{subsec:scalability}
%%%
\noindent We now demonstrate the scalability of \abacus with respect to the input graph size.  
Specifically, we measure the elapsed time that \abacus needs to fully process bipartite graph streams with varying numbers of edges. 
Note that we do not consider the waiting time for the arrival of each edge, but we only measure the actual time that \abacus needs to process the edges of a stream (with a default deletions    ratio $\alpha=20\%$). 
We use different values for the sample size $k$, i.e.,~$75K$, $150K$, and $300K$ edges.  
We measure the elapsed times each time we process another $10\%$ of edges from the entire graph stream.  

Figure~\ref{fig:scalability-input-size} illustrates the elapsed time to process the whole graph stream.  
Specifically, Figure~\ref{fig:scalability-trackers} shows that \abacus scales linearly to the input graph size for the \textit{Trackers} graph. 
As expected, a larger sample size leads to increased elapsed times; however, the linearity effect is preserved. 
In addition, we observe a similar scalability trend in the \textit{Orkut} graph stream as shown in Figure~\ref{fig:scalability-orkut}. 
Note that we received linear scalability trends on the other real bipartite graph streams, yet we omit the results for the sake of space. 
\textit{Therefore, we conclude that \abacus scales linearly to the graph input size, which is in accordance with the Theorem~\ref{theorem:time-complexity}.}

\begin{figure*}[hbtp]
  \centering
  \begin{subfigure}[t]{0.24\textwidth}
    \includegraphics[width=\textwidth]{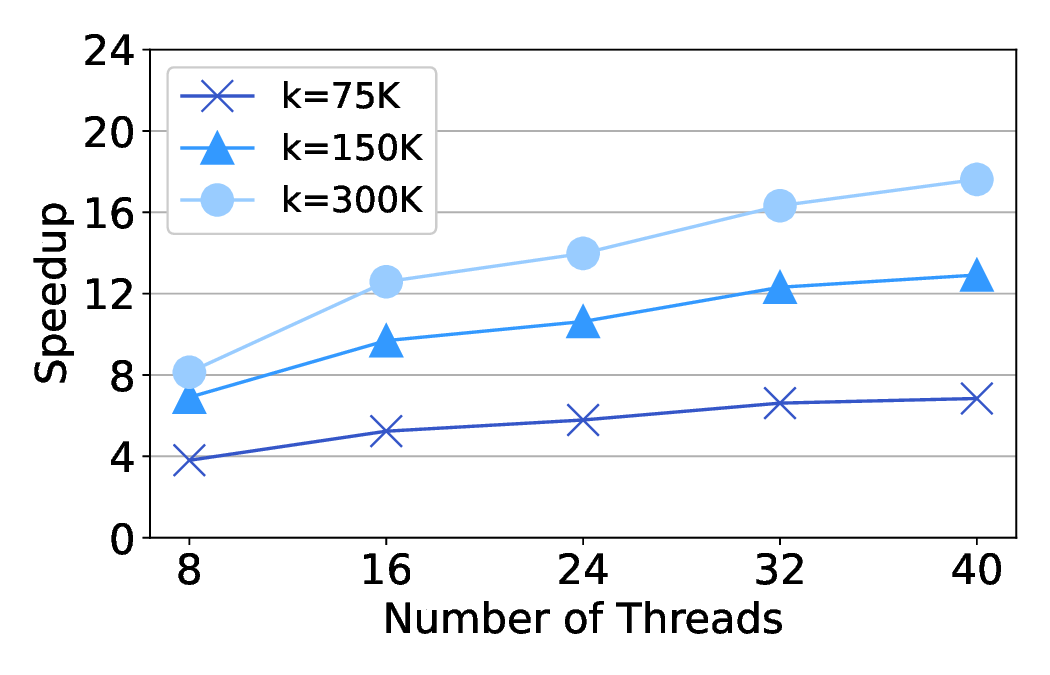}
    \caption{\textit{MovieLens}}
    \label{fig:parallelization-ex1-movielens}
  \end{subfigure}
  \begin{subfigure}[t]{0.24\textwidth} 
    \includegraphics[width=\textwidth]{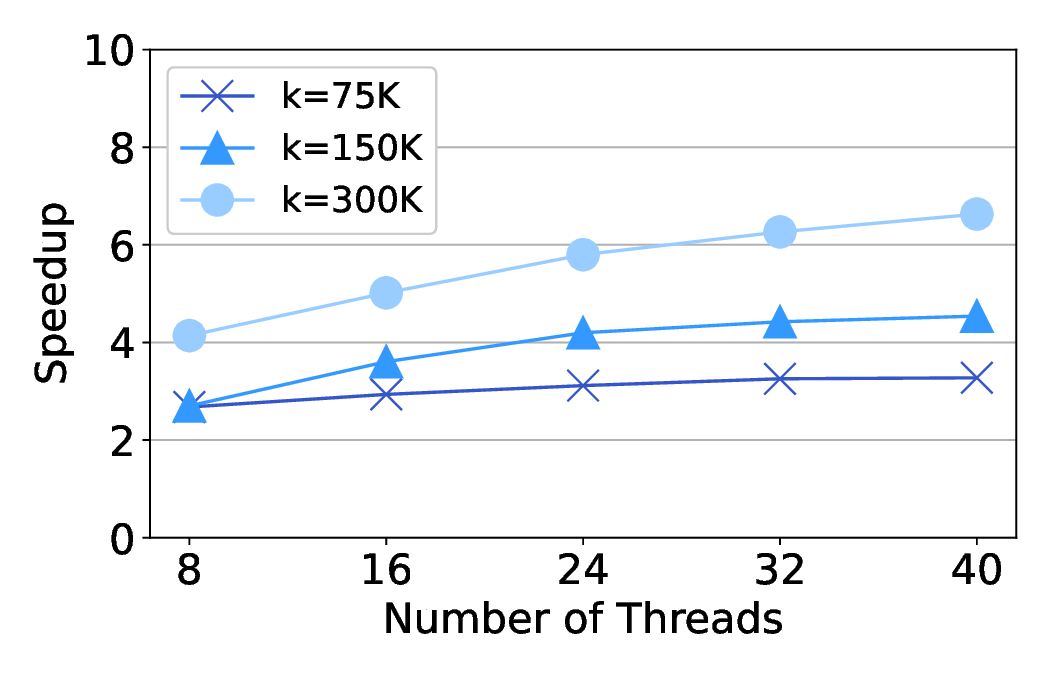}
    \caption{\textit{LiveJournal}}
    \label{fig:parallelization-ex1-livejournal}
  \end{subfigure}
  \begin{subfigure}[t]{0.24\textwidth} 
    \includegraphics[width=\textwidth]{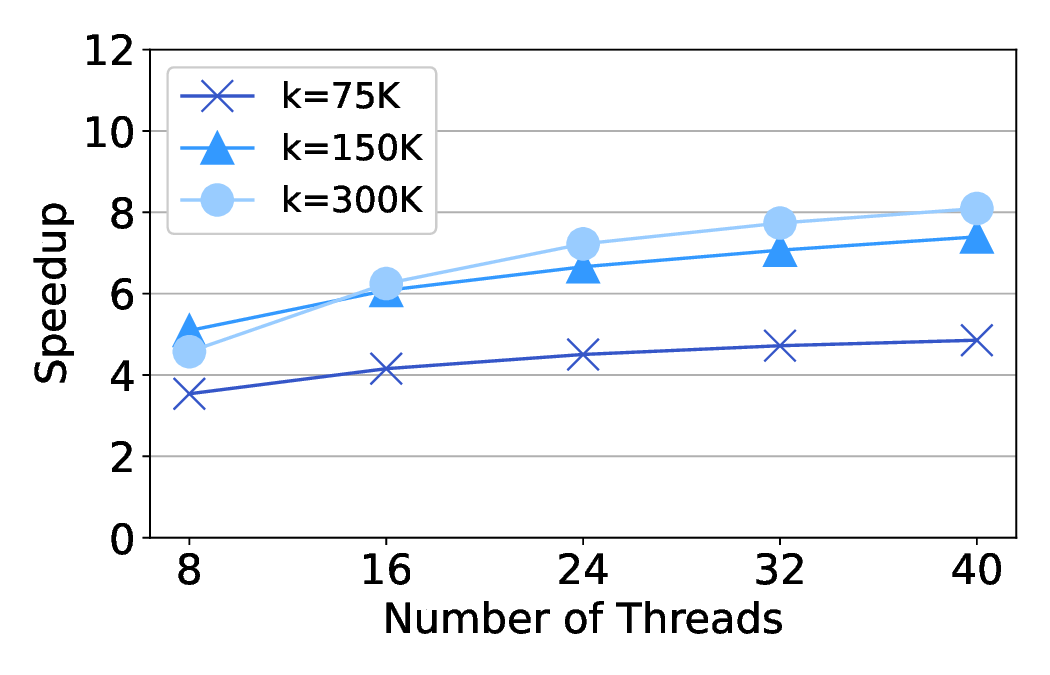}
    \caption{\textit{Trackers}}
    \label{fig:parallelization-ex1-trackers}
  \end{subfigure}
  \begin{subfigure}[t]{0.24\textwidth} 
    \includegraphics[width=\textwidth]{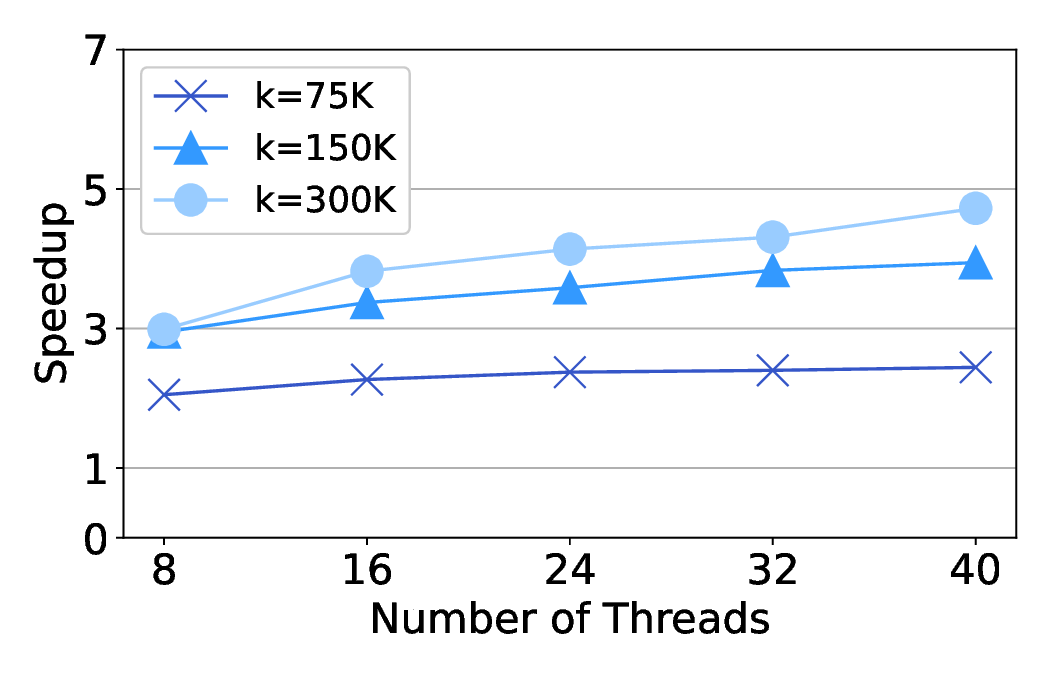}
    \caption{\textit{Orkut}}
    \label{fig:parallelization-ex1-orkut}
  \end{subfigure} 
  \caption{Speedup of \parabacus when varying the number of threads and using mini-batches of $10K$ edges. % and a fixed sample size.
  }
  \label{fig:parallelization-threads} 
\end{figure*}

\subsection{Parallelization In-depth}
\label{subsec:parallelization} 
\noindent We now analyse in depth the performance of \parabacus. 
To this end, we consider the impact of the mini-batch size and the number of threads when processing of the entire bipartite graph stream. 
Specifically, we measure the speedup in runtime that \parabacus achieves over \abacus. 

\noindent \textbf{Mini-batch Size.} Figure~\ref{fig:parallelization-minibatch} illustrates the speedup that our algorithm achieves when we vary the mini-batch size. 
We illustrate the speedup for three different sample sizes $k$, namely, $75K$, $150K$, and $300K$ edges for each dataset. 
Note that we use all $40$ available threads in this experiment. 
We see that the larger the mini-batch size, the greater the speedup we achieve.  
When the mini-batch size increases, the work assigned to each thread also increases, and consequently, parallelism is more beneficial. 
For instance, we see that for a mini-batch size of $10K$ edges, in \textit{Movielens} in Figure~\ref{fig:parallelization-ex2-movielens} we achieve up to $17.6\times$ speedup when the sample size is $300K$ edges, $12.9\times$ speedup when the sample size is $150K$ edges, and $6.84\times$ speedup when the sample size is $75K$ edges. 
In Figure~\ref{fig:parallelization-ex2-trackers} on \textit{Trackers} we achieve up to $8.1\times$ speedup when the sample size is $300K$ edges, $7.4\times$ speedup when the sample size is $150K$ edges, and $4.85\times$ speedup when the sample size is $75K$ edges.   
Interestingly, we observe that the speedup ranges that we achieve differ across different datasets.  
To validate this empirically, we additionally counted the number of vertices examined due to the set intersection operations for a sample size of $150$K for each dataset.  
We found that the total number of vertices examined was $2.21B$ in \textit{Movielens}, $0.45B$ in \textit{Livejournal}, $0.84B$ in \textit{Trackers}, and $0.30B$ in \textit{Orkut}.  
This also correlates with the density of butterflies in each dataset, with \textit{Movielens} having the highest density and \textit{Orkut} having the lowest (as shown in Table~\ref{tab:datasets}).   
As we maintain uniform random samples, the denser the graph in terms of butterfly containment, the denser the sample, and consequently, more work is done for every set intersection to identify butterflies. 
Therefore, in graphs with higher density, such as \textit{Movielens}, we observe a relatively higher speedup due to the larger workload assigned to each thread. 
Conversely, in sparser graphs like \textit{Orkut}, the speedup achieved is still significant but comparatively lower.  
Also, note that the larger the sample size, the more significant the overall speedup \parabacus achieves. 
Parallelism is more beneficial in this case because the set intersection operations related to the per-edge butterfly counting are performed between neighboring sets of a bigger size.  

\noindent \textbf{Number of Threads.} Figure~\ref{fig:parallelization-threads} shows the speedup that \parabacus achieves when we vary the number of threads for a mini-batch size fixed to $10K$ edges. 
For each dataset, we illustrate the speedup for three different sample sizes, namely, $75K$, $150K$, and $300K$ edges. 
We observe that the more threads we utilize, the greater the speedup that \parabacus attains.  
Furthermore, as the sample size increases from $75K$ to $300K$ edges, we see the workload increase, and having many threads working in parallel pays off. 
As shown in Figures~\ref{fig:parallelization-ex1-movielens}-\ref{fig:parallelization-ex1-orkut}, we achieve up to $18\times$ speedup in \textit{Movielens}, up to $6.65\times$ in \textit{Livejournal}, up to $8.1\times$ in \textit{Trackers}, and up to $5\times$ in \textit{Orkut}. 
Similar to the mini-batch experiments, the observed speedup in our experiment varies across different datasets due to their individual density characteristics. 
Consequently, the overall computation performed by the set intersections also varies across datasets. 
Additionally, the larger the sample size the bigger the performance gains are as we increase the number of threads. 
For instance, in \textit{Movielens} for a sample size of $75K$ edges as we increase the threads from $8$ to $40$ the speedup we achieve ranges from $3.8-6.85\times$, for a sample of $150K$ edges it ranges from $6.9-12.91\times$, and for sample of size $300K$ edges it ranges from $8.13-17.91\times$.

\begin{figure}[t]
  \begin{subfigure}[t]{0.235\textwidth}   
    \includegraphics[width=\textwidth]{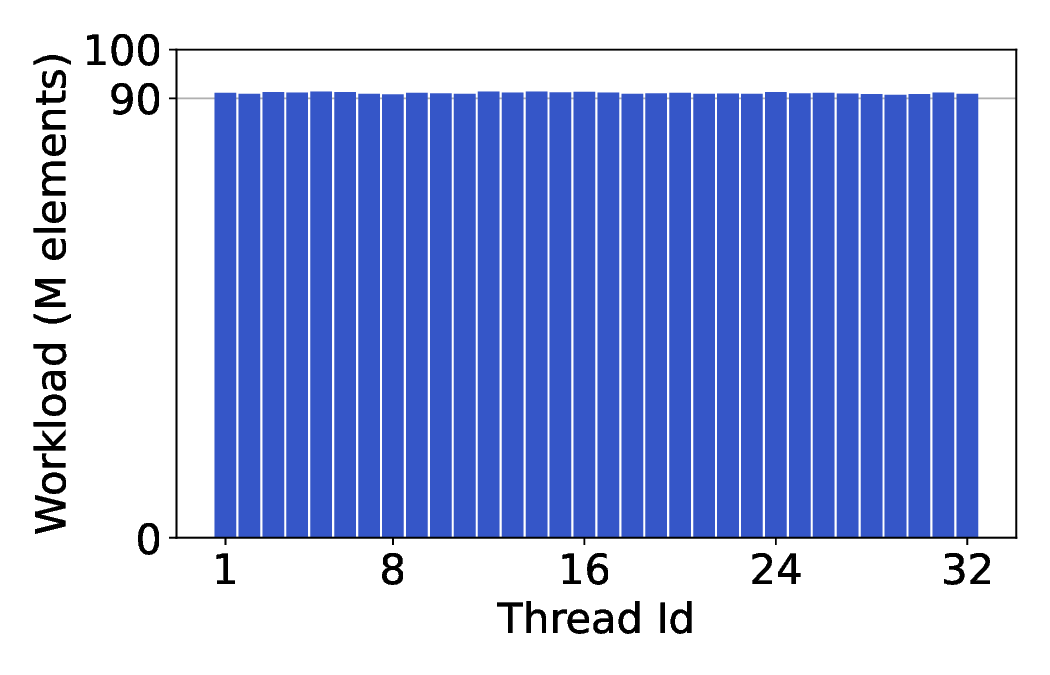}
    \caption{\textit{Movielens}}  
    \label{fig:workload-movielens}
  \end{subfigure}
    \begin{subfigure}[t]{0.235\textwidth}   
    \includegraphics[width=\textwidth]{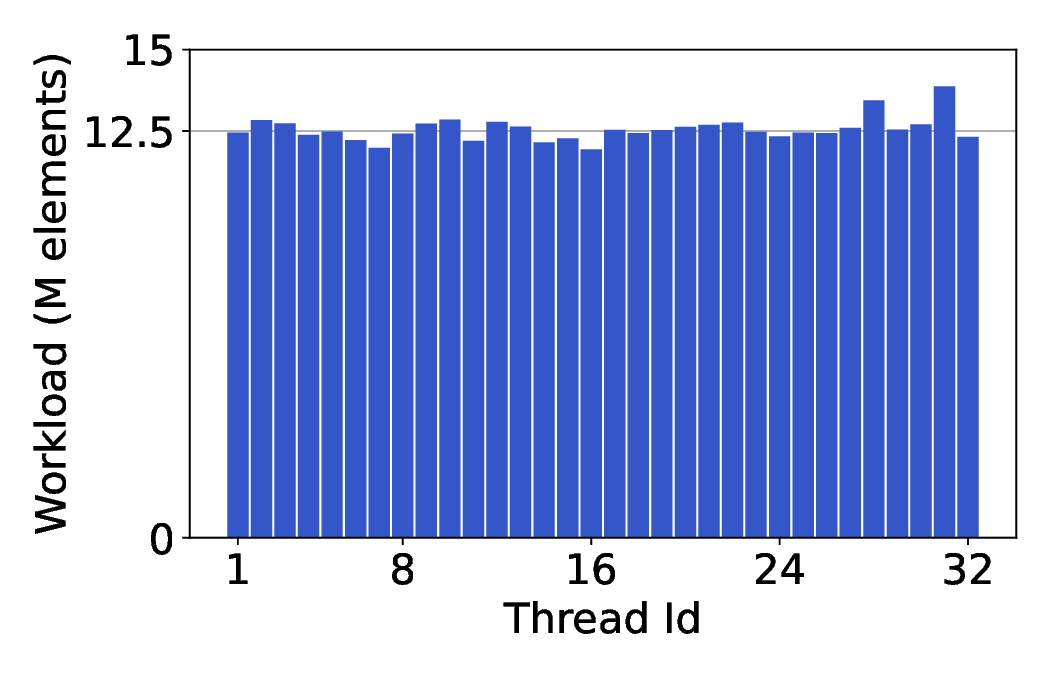}
    \caption{\textit{Orkut}}  
    \label{fig:workload-orkut}
  \end{subfigure} 
  \caption{Workload per thread.} 
  \label{fig:workload-per-thread} 
\end{figure}

\noindent \textbf{Balanced Load.} 
Figure~\ref{fig:workload-per-thread} illustrates the workload per thread, i.e.,~the number of checks that happened within the set intersection operations that take place during the butterfly counting. 
In this experiment, we use a sample size of $150$K edges, a mini-batch size of $10$K edges, and set the number of threads to $32$. 
Specifically, in Figure~\ref{fig:workload-movielens} we illustrate the workload per thread for the densest graph, \textit{Movielens}, and in Figure~\ref{fig:workload-orkut} we show the workload per thread for the sparsest graph, \textit{Orkut}, in terms of butterfly density. 
We observe that all threads are assigned similar workloads, which indicates that the computations of \parabacus are load-balanced. 
Additionally, we see that in \textit{Movielens} the average per-thread load is $90$M element comparisons, whereas in \textit{Orkut} it is $12.5$M element checks.  
This is in accordance with our previous observation that the work that is needed to process dense graphs is more than that for sparse ones. 
Same observations hold for the other datasets; however, we do not show the results for the sake of space.

\noindent \textit{We conclude that \parabacus achieves significant speedup using multi-threading, which allows for load-balanced processing of highly volatile bipartite graph streams that may receive thousands of updates per time unit.} 
%%%%%%%%%%%%%%%%%%%%%%%%%%%%%%%%%%%%%%%%%%%%%
%%%%%%%%%%%%%%%%%%%%%%%%%%%%%%%%%%%%%%%%%%%%%
\section{Related Work}
\label{sec:related-work}
%%%
% \noindent We review related work that spans triangle counting in fully dynamic graph streams, as well as counting butterflies in static graphs and in insert-only graph streams. 

\subsection{Triangles in Fully Dynamic Graph Streams} 
\label{subsec:triangle-fully-dynamic}
%%% 

\noindent Kutzkov et al.~\cite{swat14-pagh} present the first method for counting triangles in fully dynamic graph streams, which adapts colorful triangle sampling~\cite{pagh2012colorful} to obtain a sparsified graph on which the ratio of two-paths that form triangles is estimated and is afterward scaled to the whole graph. 
However, \cite{swat14-pagh} is not a real-time streaming algorithm as it only computes an estimate once at the end of the stream, and requires more memory than that for storing the whole input graph in the worst case. 
Han et al.~\cite{asonam17-esd} present ESD, which maintains the current snapshot of a fully dynamic input graph stream. 
For every incoming edge, ESD tosses a biased coin, and {\em iff} it lands on heads, it updates the triangle counts by approximating the estimate changes, rather than calculating them precisely. 
Yet, ESD is not scalable as it has to maintain the whole graph in memory. 

Triest$_{FD}$~\cite{tkdd17-triest} maintains a uniform sample 
given a specific memory budget, and derives its estimates by multiplying the triangle counts it obtains from the sampled graph and the reciprocal of the probability that each triangle is sampled.  
While Triest$_{FD}$ plainly discards the edges that are not sampled without using them for updating its count estimates, ThinkD~\cite{pkdd18-thinkd, tkdd20-thinkd} also leverages the non-sampled edges to update its triangle estimates before discarding them.

\subsection{Butterflies in Static Graphs}
\label{subsec:butterfly-static} 
%%%
\noindent Wang et al.~\cite{bigdata14-rectangle} present the first technique for exact butterfly counting in static bipartite graphs through wedge enumeration.   
Sanei-Mehri et al.~\cite{kdd18-tirthapura} improve~\cite{bigdata14-rectangle} by selecting the cheapest bipartition to traverse when computing exact butterfly counts.  
Also, the authors proposed randomized algorithms based on sampling and sparsification for computing approximate butterfly counts. 
Wang et al.~\cite{vldb19-xuemin} propose a vertex
ordering-based %-priority-based 
method, which considers the vertex degrees such that it enumerates fewer wedges %in the process of spotting the butterflies.
throughout the process of butterfly counting. 
\textsc{PARBUTTERFLY}~\cite{apocs20-shi} is a framework that conducts parallel butterfly counting with work-efficient guarantees in static bipartite graphs by exploring various vertex priority functions. 
% Shi et al.~\cite{apocs20-shi} propose a PARBUTTERFLY, a framework that counts butterfly fully in parallel relying on vertex ordering. 
Besides vertex ordering, PARBUTTERFLY also offers other type of orderings (or rankings) such as \textit{side}, \textit{approximate degree}, \textit{log-degree}, \textit{degeneracy}, %~\cite{pods20-seshadhri}, 
\textit{complement degeneracy}, and \textit{approximate complement degeneracy} orderings. 
% They prove that irrespective of the ordering used, their framework is \textit{work-efficient}, and consequently, for \textit{vertex degree} ordering that~\cite{vldb19-xuemin} uses. 
%\zoi{you also need to add what you replied to the reviewer. that in our case vertex ordering does not make sense etc. just 1 sentence should be enough} \makis{move the sentence here} 
However, adapting the vertex-ordering technique to the streaming setting is infeasible as the priorities continuously change after each incoming edge and sorting becomes a huge overhead.  
%\zoi{I just saw you have this below but we need to have a sentence here too and then shrink the text below and say similarly...} \makis{done. shrinked the blue sentences in the next subsection} 
Zhou et al.~\cite{vldb21-lei, vldbj23-lei} build on the techniques of~\cite{vldb19-xuemin} and design methods for counting butterflies in uncertain bipartite graphs. 
Xu et al.~\cite{vldb22-bingsheng} propose a GPU-based butterfly counting algorithm that uses an adaptive strategy, which balances the workload among GPU threads for maximizing efficiency. 
All the aforementioned methods are tailored to static graphs and, thus, are unsuitable for the streaming setting.

\subsection{Butterflies in Insert-Only Graph Streams}
\label{subsec:butterfly-insert-only} 
%%%
\noindent Wang et al.~\cite{vldbj22-wang} extend~\cite{vldb19-xuemin} %their method in
for dynamic graphs. 
% however, adapting the vertex-ordering technique to the streaming setting to gain performance is infeasible as the priorities continuously change after each incoming edge and sorting becomes a huge overhead.
Similarly, adapting the vertex-ordering approach %technique of~\cite{vldbj22-wang} 
to the streaming setting is infeasible as the priorities continuously change after each incoming edge and sorting becomes a huge overhead -- besides~\cite{vldbj22-wang} must store the whole graph in main memory, which is prohibitive for streaming algorithms that 
%keep only a graph sample in memory.
maintain only a sample of the graph in main memory. 
% Also,~\cite{vldbj22-wang} must have the degrees of the neighbors of all vertices sorted at all times, and thus, doing so in the streaming scenario would end up in excessive sorting, which becomes the core bottleneck. 
% Therefore, the algorithm for dynamic graphs~\cite{vldbj22-wang} cannot be used in the streaming setting. 
Sanei-Mehri et al.~\cite{cikm19-fleet} propose FLEET, %a suite of three algorithms that
which utilizes adaptive sampling for counting butterflies in bipartite %graph 
streams using a fixed memory. 
Li et al~\cite{tkde21-cas} present a Co-Affiliation Sampling (CAS) approach, which uses sampling and sketching to provide accurate estimates of butterfly counts in bipartite graph streams. 
Sheshbolouki et al.~\cite{tkdd22-sgrapp} propose sGrapp, an approximate adaptive window-based algorithm for counting butterflies after conducting a data-driven empirical analysis to reveal the temporal organizing principles of butterflies in real-time streams. 
All the above-mentioned streaming methods are tailored to insert-only bipartite graph streams. 
To the best of our knowledge, \abacus is the first method that provides accurate butterfly counts in fully dynamic bipartite graph streams. 
\section{Conclusions}
\label{sec:conclusions}
\noindent We proposed \abacus, the first algorithm that estimates butterfly counts 
in fully dynamic graph streams, which entail both insertions and deletions of edges. 
We showed that \abacus is: 
(a) accurate in estimating butterfly counts as it achieves up to $148\times$ smaller error than the baselines; 
(b), efficient as it performs butterfly counting with similar throughput as the competitors while also attaining linear scalability; and 
(c) theoretically sound as it consistently and provably provides unbiased estimates of low variance at any time as the input bipartite graph evolves. 
Additionally, we presented \parabacus, the parallel version of \abacus, which processes a graph stream in mini-batches and
counts butterflies in a load-balanced manner using versioned samples. 
We showed that \parabacus achieves considerable speedup and is thus suitable for applications in the streaming setting. 

\section*{Acknowledgments}
\label{sec:acknowledgements}
\noindent The authors would like to thank Evangelos Kipouridis for his help in completing the variance proof. 
Furthermore, we gratefully acknowledge funding from the German Federal Ministry of Education and Research under the grant BIFOLD23B.

\bibliographystyle{IEEEtran} 
\interlinepenalty=10000
\bibliography{references}

\clearpage
% \appendix
%\section*{Appendix}
%\addcontentsline{toc}{section}{Appendices}
\newpage
\appendices
\section{Proof for Lemma~\ref{lemma:discovery-probability-butterflies} (Butt. Discovery Probability)}

\begin{proof}[Proof]
\label{proof:discovery-probability-butterflies}
The probability $p^{(t)}$ that a specific set of three edges, namely, $e_1 = \{u,v\}$, $e_2 = \{w,x\}$, and $e_3 = \{y,z\}$ being selected in the sample $\mathcal{S}^{(t)}$ can be written using conditional probabilities as $P(e_1) * P(e_2 | e_1) * P(e_3 | e_1, e_2)$, where $P(e_2 | e_1)$ is the probability that edge $e_2$ exists in the sample given that edge $e_1$ already exists in the sample. 
Similarly, $P(e_3 | e_1, e_2)$ is the probability that edge $e_3$ exists in the sample given that the edges $e_1$ and $e_2$ exist in the sample. 
As the sample is uniform, the probability of containing edge $e_1$ is $P(e_1) = \frac{y}{T}$, where $y=\min (k, |E|+c_g+c_g)$ is the size of the sample, and $T=|E|+c_g+c_g$ is the size of the stream with all the edges in the stream that are not yet deleted.     
Given that edge $e_1$ exists in the sample, the probability of edge $e_2$ existing in the sample is $P(e_2|e_1) = \frac{y-1}{T-1}$ since there are only $T-1$ edges remaining after edge $e_1$ is selected, and $y-1$ of them can be selected to form the sample.  
Similarly, given that edges $e_1$ and $e_2$ exist, the probability of edge $e_3$ existing in the sample is $P(e_3|e_1,e_2) = \frac{y-2}{T-2}$. 
Therefore, Equation~\ref{eq:reciprocal} represents the probability of selecting a uniform random sample of $y$ edges from the graph stream containing the edges $e_1$, $e_2$, and $e_3$.  
\end{proof}

\section{Proof of Theorem~\ref{theorem:unbiasedness} (Unbiasedness)}
\label{appendix:proof-unbiasedness}

\begin{proof}[Proof] 
\label{proof:unbiasedness}
Let us assume a butterfly $(\{u,v,w,x\}, s) \in \mathcal{C}^{(t)}$, 
for an insertion $e^{(s)} = (\{u,v\}, +)$ without loss of generality, where $\mathcal{C}^{(t)}$ is the set of created butterflies at time $t$ or earlier.  
Furthermore, the increment amount of change, $c_{uvwx}^{(s)}$, in the butterfly count 
$c$ that is attributed to the creation of the butterfly ($\{u,v,w,x\}, s$) (Algorithm~\ref{alg:abacus}, lines~6,11) is as follows:
\begin{equation}
\label{eq:amount_of_change_unbiasedness_insertion}
\hspace{-0.2cm}
    c_{uvwx}^{(s)} = \begin{cases}
        \frac{1}{p^{(s-1)}} = \frac{1}{Pr(\{u,x\} \in \mathcal{S}^{(s-1)} \cap \{v,w\} \in \mathcal{S}^{(s-1)} \cap \{w,x\} \in \mathcal{S}^{(s-1)})} \\ \text{, if 
        $\{u,x\}$, $\{v,w\}$, and $\{w,x\} \in \mathcal{S}^{(s-1)}$} \\ $0$\text{, otherwise} \nonumber
    \end{cases} 
\end{equation}
\noindent which holds from Equation~\ref{eq:discovery-probability-butterflies}. 
Therefore, it trivially holds that $\mathop{\mathbb{E}}(c_{uvwx}^{(s)}) = 1$. 
Similarly, assume a butterfly $(\{u,v,w,x\}, s) \in \mathcal{D}^{(t)}$, 
for an edge deletion $e^{(s)} = (\{u,v\}, -)$ without loss of generality, where $\mathcal{D}^{(t)}$ is the set of deleted butterflies at time $t$ or earlier. 
Furthermore, the amount of change, $d_{uvwx}^{(s)}$, in the butterfly count $c$ that is attributed to the deletion of the butterfly ($\{u,v,w,x\}, s$) (Algorithm~\ref{alg:abacus}, lines~6,11) is as follows:
\begin{equation}
\label{eq:amount_of_change_unbiasedness_deletion}
\hspace{-0.2cm}
    d_{uvwx}^{(s)} = \begin{cases}
        \frac{-1}{p^{(s-1)}} = \frac{-1}{Pr(\{u,x\} \in \mathcal{S}^{(s-1)} \cap \{v,w\} \in \mathcal{S}^{(s-1)} \cap \{w,x\} \in \mathcal{S}^{(s-1)})} \\ \text{, if 
        $\{u,x\}$, $\{v,w\}$, and $\{w,x\} \in \mathcal{S}^{(s-1)}$} \\ $0$\text{, otherwise} \nonumber
    \end{cases} 
\end{equation}
\noindent which holds from Equation~\ref{eq:discovery-probability-butterflies}. 
Thus, it holds that $\mathop{\mathbb{E}}(d_{uvwx}^{(s)}) = -1$.  
Therefore, for the butterfly count, $c$, it holds that:
\begin{equation}
\label{eq:def_summation_global_count} 
    c^{(t)} = \sum_{(\{u,v,w,x\}, s) \in \mathcal{C}^{(t)}} c_{uvwx}^{(s)} + \sum_{(\{u,v,w,x\}, s) \in \mathcal{D}^{(t)}} d_{uvwx}^{(s)} \nonumber
\end{equation}
\noindent By linearity of expectation the following equality holds:
\begin{align}
\label{eq:def_lin_exp_global_count}
    \mathop{\mathbb{E}}(c^{(t)}) &= \sum_{(\{u,v,w,x\}, s) \in \mathcal{C}^{(t)}} \mathop{\mathbb{E}}(c_{uvwx}^{(s)}) + \sum_{(\{u,v,w,x\}, s) \in \mathcal{D}^{(t)}} \mathop{\mathbb{E}}(d_{uvwx}^{(s)}) \nonumber \\
    &= \sum_{(\{u,v,w,x\}, s) \in \mathcal{C}^{(t)}} (+1) + \sum_{(\{u,v,w,x\}, s) \in \mathcal{D}^{(t)}} (-1) \nonumber \\ &= |\mathcal{C}^{(t)}| - |\mathcal{D}^{(t)}| = |B^{(t)}|
\end{align} 
\end{proof}

\section{Proof of Theorem~\ref{theorem:variance} (Variance)}
\label{appendix:proof-variance}

\begin{proof}[Proof] \label{proof:variance}
As we described in Algorithm~1, \abacus maintains a uniform random sample at all times, which has size $0 \leq |\mathcal{S}| \leq k$, where $k$ is our memory budget or the maximum sample size. 
In the beginning of the bipartite graph stream \abacus's sample is empty and \abacus keeps the first $k$ edges (attaining zero variance and conducting exact estimation). 
Later on, \abacus always maintains a sample that has size $k$. 
The only case where the sample size will get smaller than $k$ is when an edge deletion arrives and that edge is specifically stored in $\mathcal{S}$, which is later compensated directly by the subsequent insertions such that the sample size becomes $k$ again. 
Therefore, we assume that the sample $\mathcal{S}$ has $k$ edges. 
It shall be clear from the proof that our analysis is robust and would give nearly the same bounds even if $\mathcal{S}$ has size close to $k$ but not exactly $k$.

Consider a particular moment in time denoted by $t$, where the sample $\mathcal{S}$ contains $k$ edges. 
At the given time $t$, let $|B^{(t)}|$ denotes the true count of butterflies (the ground truth) observed in the graph stream so far, while $|B_{\mathcal{S}}^{(t)}|$ represents the number of butterflies present in the sample. 
In the sequel, we omit the time $t$ notations from our formulas for the sake of clarify. 
% Let a specific point in time with timestamp $t$, where the sample $\mathcal{S}$ contains $k$ edges, and the number of butterflies in the sample are $|B_{\mathcal{S}}|$, whereas the true number of butterflies (our ground truth) in the whole graph stream is $|B|$. 
Figure~\ref{fig:instance-variance-proof} illustrates an instance of sample $\mathcal{S}$, where the edges that appear in the graph stream are depicted with black color, while the edges that are present in the sample $\mathcal{S}$ are highlighted in red. % colour. 
% In Figure~\ref{fig:instance-variance-proof}, we illustrate an instance of sample $\mathcal{S}$ where we illustrate the edges that appeared in the graph stream (with black color), and the edges that reside in the sample $\mathcal{S}$ (with red colour). 
We see that at time $t$ the true number of butterflies is $3$, i.e.,~$B_1, B_2$, and $B_3$, but using the red edges in the sample we count only $1$, namely, $B_1$. 
Therefore, we observe that through the sample $\mathcal{S}$, \abacus can either spot a formed butterfly pattern or not. 
To facilitate our variance analysis, we introduce \textit{indicator random variables} $X_i$, which equals to $1$ if the butterfly $B_i \in \mathcal{S}$, and $0$ if $B_i \notin \mathcal{S}$. 
Subsequently, the total number of butterflies that \abacus spots in the sample $\mathcal{S}$ is $|B_{\mathcal{S}}| = \sum_{i} X_i$, whereas after extrapolating (as we described in Algorithm~\ref{alg:abacus}) the actual butterfly count that \abacus returns is $c = \gamma \sum_{i} X_i$, where %$inc$
% is the \textit{increment} that \abacus utilizes to extrapolate $|B_{\mathcal{S}}|$ (Alg.~\ref{alg:abacus}, lines~6,11). 
$\gamma$ is equal to $\frac{groundTruth}{groundTruth \times \binom{E-4}{k-4} / \binom{E}{k}} = \binom{E}{k} / \binom{E-4}{k-4}$ since we already proved in Theorem~\ref{theorem:unbiasedness} that \abacus is an unbiased estimator.   
% \kipou{For me it's perfectly clear... but to be fair, $c$ is equal to what we claim, but it's not how it's computed in the algorithm... hmm... maybe you can straight away say that since your algorithm is an unbiased estimator, then \textit{inc} is equal to groundTruth / expected-sum-of-$X_i$, which is $groundTruth / (groundTruth * \binom{E-4}{k-4} / \binom{E}{k}) = \binom{E}{k} / \binom{E-4}{k-4}$} \makis{\textcolor{red}{FIX ME}}
From the definition of variance, the following equality holds: 
\begin{align}
\label{eq:variance-proof-seq}
    Var[c] &= \mathbb{E}[(c - \mathbb{E}[c])^2] \nonumber \\
    &= \mathbb{E}[c^2] - \mathbb{E}[c]^2 
\end{align} 
In Theorem~\ref{eq:unbiasedness}, we showed that 
% $\mathbb{E}[c^{(t)}] = \mathbb{E}[c]$ 
$\mathbb{E}[c^{(t)}] = |B^{(t)}|$ 
is unbiased.  
Therefore, we must calculate the other portion of Equation~\ref{eq:variance-proof-seq}. 
In specific, we unfold the $\mathbb{E}[c]^2$ as follows: 
% \kipou{I would completely skip the line with the dots, seems ugly. I would replace it with the following: $inc^2 (\mathbb{E}[\sum_i X_i^2]) + 2 inc^2 (\mathbb{E}[\sum_{i,j} X_iX_j])$ (this is basically the next line, but with $X_i^2$ in the first sum... and the next line is converting it to $X_i$ because $X_i$ is either $0$ or $1$ and thus $X_i^2 = X_i$... } \kipou{you can use align* (with an asterisk) and then you don't need to say nonumber, it automatically does it}
\begin{align*}
\label{eq:variance-proof-seq-2}
    \mathbb{E}[c^2] &= \gamma^2 \mathbb{E}[(\sum_i X_i)^2] \\ 
    % &= inc^2 ( \mathbb{E}[X_1^2] + \mathbb{E}[X_2^2] + \dots + \mathbb{E}[X_{|B_{\mathcal{S}}|}] + \dots + \nonumber \\
    % &= inc^2\mathbb{E}[\sum_i X_i^2] + 2 inc^2 \mathbb{E}[\sum_{i,j} X_iX_j] \\
    % \mathbb{E}[2 X_1 X_2] &+ \mathbb{E}[2 X_1 X_3] + \dots + \mathbb{E}[2 X_{|B_{\mathcal{S}}| -1} X_{|B_{\mathcal{S}}|}] ) \\
    % &= inc^2 (\mathbb{E}[\sum_i X_i]) + 2 inc^2 (\mathbb{E}[\sum_{i,j} X_iX_j]) \\ 
    &= \gamma (\mathbb{E}[\gamma \sum_i X_i]) 
    + 2 \gamma^2 (\mathbb{E}[\sum_{i,j} X_iX_j]) \\ 
    &= \gamma (\mathbb{E}[c]) 
    + 2 \gamma^2 (\mathbb{E}[\sum_{i,j} X_iX_j]) \\ 
    &= \gamma (\mathbb{E}[c]) 
    + 2 \gamma^2 \sum_{i,j} \mathbb{E}[X_iX_j] \\ 
    &= \gamma (\mathbb{E}[c]) 
    + 2 \gamma^2 \sum_{i,j} p_{i,j} 
\end{align*} 
\noindent where $p_{i,j} =\mathbb{E}[X_iX_j]$ is the probability that the graph sample $\mathcal{S}$ contains both the butterfly $B_i$ and the $B_j$. 
In order to calculate the probability $p_{i,j}$ we distinguish the following three cases as shown in Figure~\ref{fig:pij_probability_cases}.  
Specifically, Figure~\ref{fig:pij_probability_cases}i describes the case where the butterflies $B_i$ and $B_j$ exist in the sample $\mathcal{S}$, but they do not share any edge. 
Figure~\ref{fig:pij_probability_cases}ii describes the case where the butterflies $B_i$ and $B_j$ exist in the sample $\mathcal{S}$, and they share one edge. 
Figure~\ref{fig:pij_probability_cases}i describes the case where the butterflies $B_i$ and $B_j$ exist in the sample $\mathcal{S}$, and they share two edges. 
Notice that it is not possible for two distinct 
butterflies to share three edges, because then they would be the same butterfly. 
\begin{figure}[t]
\centering 
\includegraphics[width=0.44\textwidth]{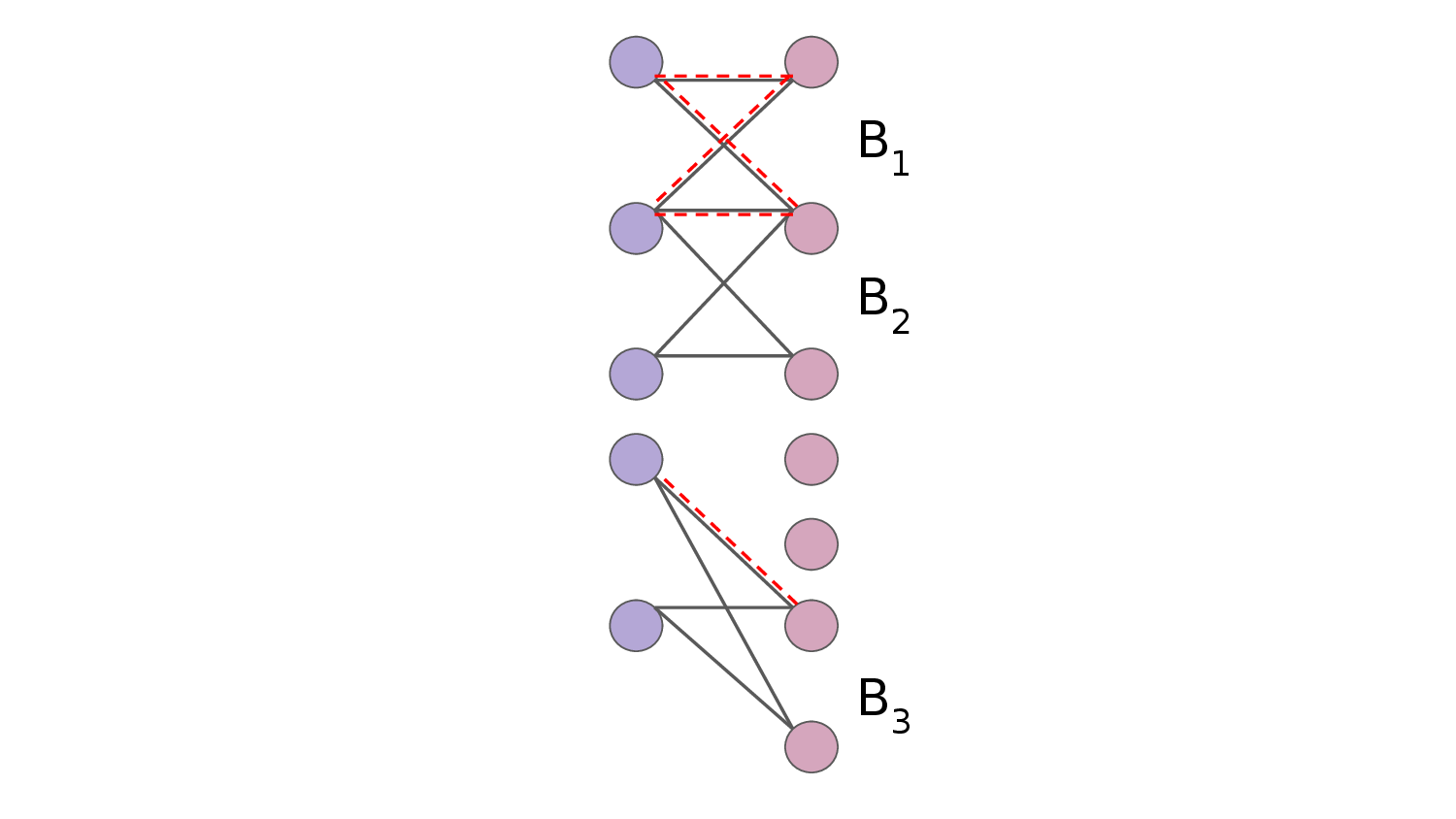} 
\caption{Instance of the graph stream (black edges) and of the graph sample (red edges) \abacus maintains at timestamp $t$.} 
\label{fig:instance-variance-proof} 
\end{figure}
The probability that \abacus encounters case~i is $\frac{\binom{|E|-8}{k-8}}{\binom{|E|}{k}}$, where $|E|$ is the number of edges that are still valid in the graph stream (have not been deleted), and $k$ is the memory budget of \abacus that indicates the maximum number of edges in the maintained graph sample $\mathcal{S}$. 
Similarly, the probability that \abacus encounters case~ii is $\frac{\binom{|E|-7}{k-7}}{\binom{|E|}{k}}$, and the probability that it encounters case~iii is $\frac{\binom{|E|-6}{k-6}}{\binom{|E|}{k}}$. 
The above probabilities stem from the fact that \abacus always maintains a uniform random sample. 
% \kipou{maybe you can say that all these follow from the fact that our sample is uniformly random}. 
Consequently, the probability of the equation above is as follows: 
\begin{align} 
    \sum_{i,j} p_{i,j} &= y_1 \frac{\binom{|E|-8}{k-8}}{\binom{|E|}{k}} + y_2 \frac{\binom{|E|-7}{k-7}}{\binom{|E|}{k}} + y_3 \frac{\binom{|E|-6}{k-6}}{\binom{|E|}{k}} \nonumber
\end{align}
\noindent where $y_1, y_2, y_3$ indicate how many pairs of butterflies of case~$i$,~$ii$,~$iii$ exist in the graph at the current time $t$ (which we omit for simplicity).  
By combining the two equations above, we get the following closed form for the variance of \abacus: 
\begin{align} 
    Var[c] &= \mathbb{E}[c^2] - \mathbb{E}[c]^2 \nonumber \\
    = \gamma \mathbb{E}[c] &+ 2 \gamma^2 \sum_{i,j} p_{i,j} - \mathbb{E}[c]^2 \nonumber \\ 
    = \gamma \mathbb{E}[c] & - \mathbb{E}[c]^2 \nonumber \\ 
    + 2 \gamma^2 (y_1 \frac{\binom{|E|-8}{k-8}}{\binom{|E|}{k}} &+ y_2 \frac{\binom{|E|-7}{k-7}}{\binom{|E|}{k}} + y_3 \frac{\binom{|E|-6}{k-6}}{\binom{|E|}{k}} ) \nonumber 
\end{align}
\noindent where by upper bounding %the sum of probabilities 
$\sum_{i,j} p_{i,j}$ as follows: 
% \kipou{maybe: $\sum_{i,j} p_{i,j} \leq (y_1+y_2+y_3) \times \frac{\binom{|E|-6}{k-6}}{\binom{|E|}{k}}  = \binom{E[c]}{2} \times \frac{\binom{|E|-6}{k-6}}{\binom{|E|}{k}}$}
\begin{align*} 
    \sum_{i,j} p_{i,j} \leq (y_1+y_2+y_3) \times \frac{\binom{|E|-6}{k-6}}{\binom{|E|}{k}} = \binom{E[c]}{2} \times \frac{\binom{|E|-6}{k-6}}{\binom{|E|}{k}} 
    % \sum_{i,j} p_{i,j} \leq \binom{E[c]}{2} \times \frac{\binom{|E|-6}{k-6}}{\binom{|E|}{k}} 
\end{align*}
\noindent we can upper bound the %portion of Equation~6
the variance as follows:
\begin{align} 
    Var[c] \leq \gamma \mathbb{E}[c] \nonumber  
    + 2 \gamma^2 \binom{E[c]}{2} \times \frac{\binom{|E|-6}{k-6}}{\binom{|E|}{k}} - \mathbb{E}[c]^2 
\end{align}
\noindent where $c$ is the butterfly count estimation of \abacus, whose expected value equals the ground truth, $|E|$ is the number of valid edges in the graph stream that have not yet been deleted, and $k$ is the memory budget or equivalently the maximum number of edges in the sample $\mathcal{S}$ that \abacus maintains. 
We observe that \abacus' variance is bounded, and it provably gives accurate estimates. 
\begin{figure}[t]
\centering 
\includegraphics[width=0.45\textwidth]{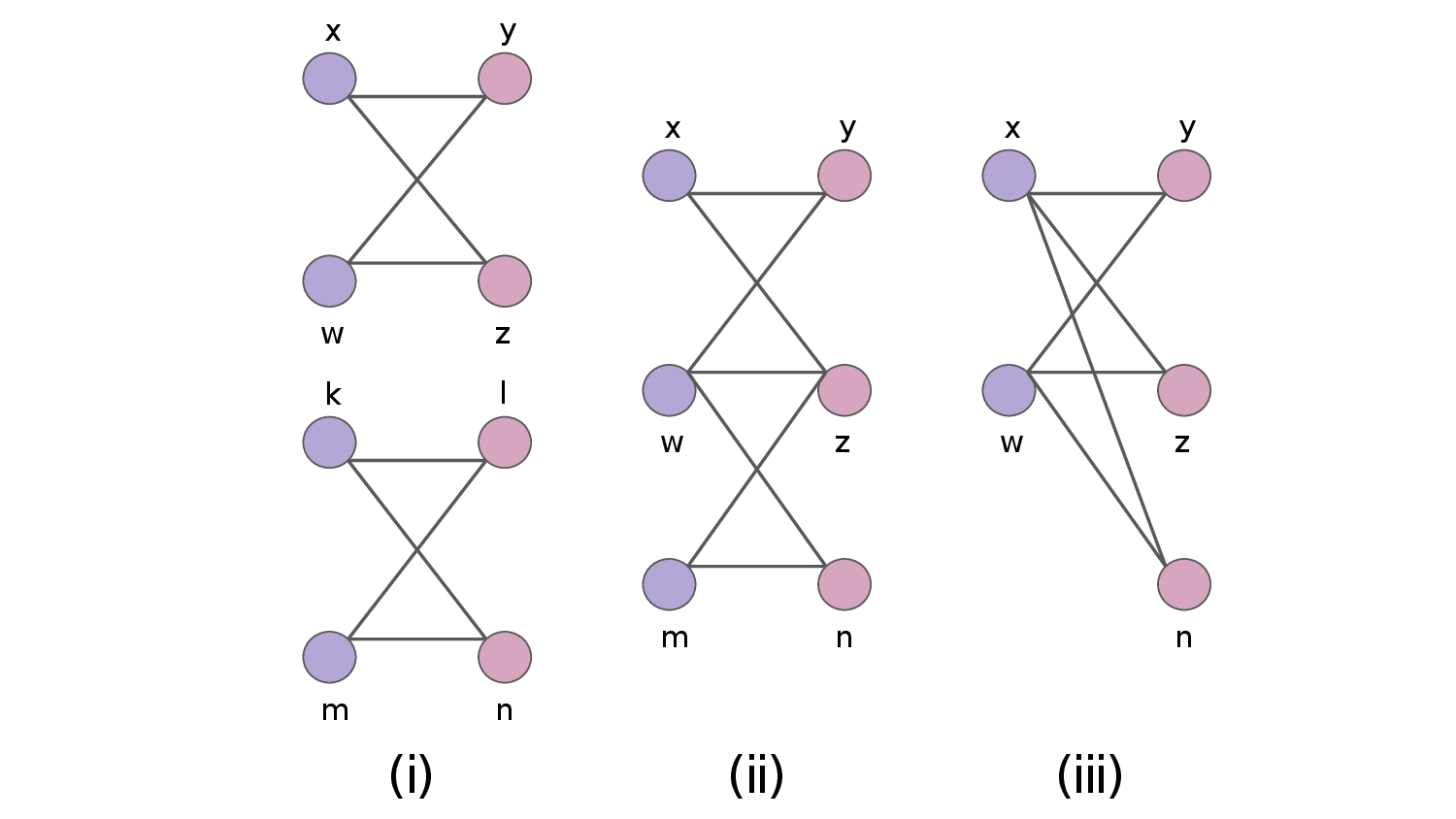} 
\caption{Cases where the probability $p_{i,j}$ is non-zero.} 
\label{fig:pij_probability_cases} 
\end{figure}
\noindent \textit{Tightness.} The upper bound we derived on the variance of \abacus is \textit{tight}.  
In specific, consider the $2,3$-bipartite graph that is a clique. 
There, the variance satisfies the above inequality 
%~\ref{ineq:variance-bound} 
on the equality condition. 
\end{proof}

\section{Proof of Corollary~\ref{corollary:concentration} (Concentration)}
\label{appendix:proof-concentration}

\begin{proof}[Proof] \label{proof:concentration} 
The proof stems directly from applying the Chebysev's inequality since we know the variance (we proved a closed form) and the expected value of \abacus's estimation. 
\end{proof} 

\end{document}